\pdfoutput=1
\documentclass[a4paper]{scrartcl}
\usepackage[english]{babel}
\usepackage{microtype}
\usepackage{amsmath, amssymb, amsthm, bbm, nicefrac}
\usepackage{xspace}
\usepackage{mathrsfs}
\usepackage{varioref}
\usepackage[shortlabels]{enumitem}
\usepackage{color}
\usepackage[color, arrow, matrix, curve, pdftex]{xy} 
	\CompileMatrices
\usepackage{url}
\usepackage{csquotes}
\usepackage{tikz}
	\usetikzlibrary{decorations.pathmorphing}
	\usetikzlibrary{arrows,patterns}
\usepackage{hyperref}

\title{Invariant states of quantum birth and death chains}
\author{
	David~B\"ucher\\[1.4ex]%
	\begin{small}RTG \enquote{Mathematical Structures in Modern Quantum Physics},\end{small}\\[-1.2ex]%
	\begin{small}Universit\"at G\"ottingen, Bunsenstr. 3--5, D--37073 G\"ottingen, Germany\end{small}\\%
	\begin{small}E-mail address: dbuecher@uni-math.gwdg.de\end{small}\\%
}
\date{August 26, 2014}

\newcommand*{\hilb}{\mathscr}
\newcommand*{\mb}{\mathbf}

\newcommand*{\N}{\mathbb N}
\newcommand*{\Z}{\mathbb Z}
\newcommand*{\R}{\mathbb R}
\newcommand*{\C}{\mathbb C}
\renewcommand*{\O}{\mathcal O}

\newcommand*{\one}{\mathbbm 1}

\newcommand{\ie}{\mbox{i.e.}\xspace}

\renewcommand*{\H}{\hilb H}

\newcommand*{\B}{\hilb B}

\DeclareMathOperator{\UCP}{UCP}
\DeclareMathOperator{\Tr}{Tr}
\DeclareMathOperator{\spann}{span}
\DeclareMathOperator{\supp}{supp}

\newcommand{\BH}{\B(\H)}
\newcommand*{\ltwo}{\ell^2(\N_0)}
\newcommand*{\Bl}{\B(\ell^2(\N_0))}

\newcommand{\diag}{\text{diag}}
\newcommand{\id}{\textnormal{id}}

\newcommand{\skapro}[1]{\langle {#1} \rangle}

\newcommand{\Dom}{\textnormal{Dom}}

\theoremstyle{plain}
	\newtheorem{lemma}{Lemma}
	\newtheorem{prop}[lemma]{Proposition}
	\newtheorem{thm}[lemma]{Theorem}
	
\theoremstyle{definition}
	\newtheorem{defn}[lemma]{Definition}
\theoremstyle{remark}

\begin{document}

\maketitle

\begin{abstract}
A sufficient condition is given 
for a class of quantum birth and death chains on the 
non-negative integers to possess invariant states. 
The result is applied to generalised one-atom masers 
and to the Jaynes-Cummings one-atom maser 
with random interaction time and not necessarily 
diagonal atomic states. 
\end{abstract}

\subsection{Introduction}

A classical birth and death chain on the non-negative integers 
is a homogeneous discrete Markov process $(x_t)_{t\in\Z}$ 
on the state space $\N_0 = \{0,1,2,\ldots\}$, where only 
nearest-neighbour-transitions occur. Such a 
process is determined by its birth rates 
$\lambda_n = P(x_t = n+1 \,|\, x_{t-1} = n)$, for $n \in \N_0$, 
and death rates 
$\mu_n = P(x_t = n-1 \,|\, x_{t-1} = n)$, for $n \in \N$. 
Suppose that $\mu_n \neq 0$ for all $n \in \N$, and 
put 
$\pi_n := \frac{\lambda_0\lambda_1\cdots\lambda_{n-1}}{\mu_1\mu_2\cdots\mu_n}$, 
$\pi_0 := 1$. 
Then $(x_t)$ has an invariant state if and only if the 
sequence $(\pi_n)_n$ is summable. In that case, the 
invariant state is unique and its density is given 
by $\rho_n = \frac{\pi_n}{\sum_n \pi_n}$. 
In particular, $(x_t)$ has an invariant state 
whenever there exists a constant $c$ with 
$\frac{\lambda_n}{\mu_{n+1}} \leq c < 1$ for 
almost all $n \in \N$. 

In the present paper, we consider invariant states of 
quantum birth and death chains (QBDCs). 
For us, a QBDC is a quantum Markov chain on $\Bl$ in the sense of 
a unital completely positive (ucp) map $T: \Bl \to \Bl$, 
with the extra-condition that only nearest-neighbour transitions are allowed. 
This means that the transition rates 
$\Tr(e_{n,m}T(e_{k,l}))$, where $e_{n,m}, e_{k,l} \in \Bl$ 
denote matrix units, vanish whenever $|m-k|>1$ or 
$|n-l|>1$. 
Examples of QBDCs are provided by the one-atom maser as 
studied in, e.g. \cite{bruneau09, bruneau13}, 
and by its generalisations considered in 
\cite{BGKRSS13}. 
QBDCs are fundamentally different from 
\enquote{unitary} or 
\enquote{open quantum random walks} as 
studied in, e.g. \cite{konno02} or \cite{attal12}. 
States for such quantum random walks not only specify 
a \enquote{position} on $\N$ (or $\Z$), but in addition 
specify the state of a \enquote{coin}. 
Open quantum random walks where the space of states of 
the coin is trivial (i.e. $\C$) are in fact 
classical Markov chains. 

{\sloppy
In addition to the analogues of the classical birth and death rates 
$\lambda_n = \Tr(e_{n,n}T(e_{n+1,n+1}))$ and 
$\mu_n = \Tr(e_{n,n}T(e_{n-1,n-1}))$, 
a QBDC is characterised, among others, by the transition rates 
$\eta_n = \Tr(e_{n+1,n}T(e_{n,n}))$, see Figure \ref{fig:transition-rates}. 
In this essentially two-dimensional setting it seems no more 
possible to give an explicit formula for an invariant state, 
as it was in the classical case. 
Moreover, the condition $\frac{\lambda_n}{\mu_{n+1}} \leq c < 1$ 
is no more sufficient to guarantee the existence of 
an invariant state (see Section \ref{sec:generalised-maser} for an example). 
However, a simple sufficient criterion can be given which 
involves the parameters $\lambda_n,\mu_n$ and $\eta_n$ 
only. Namely, if $\kappa := \lim\inf_n \frac 1n \ln(\pi_n^{-1}) > 0$ 
and 
\begin{align}
	\lim\inf_n \frac{\lambda_n\mu_{n+1}}{4|\eta_n|^2} 
		> \frac{e^{-\kappa}}{(1 - e^{-\kappa})^2} \ , 
\end{align}
then $T$ has an invariant state (Theorem \ref{thm:main-result}). 
The proof combines ideas from \cite{FaR01} 
with a positivity criterion for tri-diagonal 
matrices. 
}

While previous works have often focused 
on QBDCs with $\eta_n = 0$ for $n \gg 0$, our result is 
suitable for QBDCs with $\eta_n \neq 0$. 
We apply it to the class of QBDCs, inspired by the micromaser 
experiment, which were considered in \cite{BGKRSS13}, 
and to the Jaynes-Cummings one-atom maser with random interaction time. 
In the latter case it is shown, in particular, that 
$\frac{\lambda_n}{\mu_{n+1}} \leq c < 1$ 
suffices to guarantee the existence of invariant states even for 
non-thermal (non-diagonal) and non-pure atomic states. 
Moreover, for non-pure atomic states, every initial state 
approaches the invariant state (Proposition \ref{prop:random-JC}).

\medskip

\emph{Acknowledgements:} 
I thank Andreas G\"artner, Walther Reu{\ss}wig, Kay Schwieger and 
Florian Steinberg for useful comments on an early draft of this note, 
and Prof. Burkhard K\"ummerer for his generous support. 
Much of this research was done while the author was 
affiliated with the Fachbereich Mathematik of 
Technische Universit\"at Darmstadt.

\begin{figure}[bht]
	\[
	\begin{xy}
		\xymatrix@C=.7pc@R=.7pc{
			\bullet
			\ar@/_/[rrrrdddd]|{\mu_1}
			\ar@/_/[rrrr]|{\eta_0}
			\ar@/_/[dddd]|{\overline{\eta_0}}
			\ar@(l,u)|(.7){\sigma_0}
			& & & & \bullet
			\ar@[gray]@/_/[rrrrdddd]|(.6){}
			\ar@[gray]@/_/[lllldddd]|(.6){}
			\ar@[gray]@/_/[llll]_{}
			\ar@[gray]@/_/[rrrr]_{}
			\ar@[gray]@/_/[dddd]_{}
			\ar@[gray]@(lu,ru)^{}
			& & & & \bullet
			\ar@[gray]@/_/[rrrrdddd]|(.6){}
			\ar@[gray]@/_/[lllldddd]|(.6){}
			\ar@[gray]@/_/[llll]_{}
			\ar@[gray]@/_/[rrrr]_{}
			\ar@[gray]@/_/[dddd]_{}
			\ar@[gray]@(lu,ru)^{}
			& & & & \bullet
			\ar@[gray]@/_/[ddddllll]_{}
			\ar@[gray]@/_/[llll]_{}
			\ar@[gray]@/_/[dddd]_{}
			\ar@[gray]@(lu,ru)^{}
			& \cdots
			\\ \\ 
			\\ \\
			\bullet
			\ar@[gray]@/_/[rrrrdddd]|(.6){}
			\ar@[gray]@/_/[rrrruuuu]|(.6){}
			\ar@[gray]@/_/[rrrr]_{}
			\ar@[gray]@/_/[uuuu]_{}
			\ar@[gray]@/_/[dddd]_{}
			\ar@[gray]@(ld,lu)^{}
			& & & & \bullet
			\ar@/_/[rrrrdddd]|{\mu_2}
			\ar@/_/[uuuullll]|{\lambda_0}
			\ar@/_/[llll]|{-\overline{\eta_0}}
			\ar@/_/[rrrr]|{\eta_1}
			\ar@/_/[uuuu]|{-\eta_0}
			\ar@/_/[dddd]|{\overline{\eta_1}}
			\ar@`{+(9,4),+(-2,9)}|(.56){\sigma_1}
			& & & & \bullet
			\ar@[gray]@/_/[rrrrdddd]|{}
			\ar@[gray]@/_/[uuuullll]_{}
			\ar@[gray]@/_/[uuuurrrr]_{}
			\ar@[gray]@/_/[lllldddd]|(.6){}
			\ar@[gray]@/_/[llll]_{}
			\ar@[gray]@/_/[rrrr]_{}
			\ar@[gray]@/_/[uuuu]_{}
			\ar@[gray]@/_/[dddd]_{}
			\ar@[gray]@`{+(-9,-4),+(2,-9)}_(.56){}
			& & & & \bullet
			\ar@[gray]@/_/[uuuullll]_{}
			\ar@[gray]@/_/[lllldddd]|(.6){}
			\ar@[gray]@/_/[llll]_{}
			\ar@[gray]@/_/[uuuu]_{}
			\ar@[gray]@/_/[dddd]_{}
			\ar@[gray]@`{+(9,4),+(-2,9)}_(.56){}
			& \cdots
			\\ \\ 
			\\ \\
			\bullet
			\ar@[gray]@/_/[rrrrdddd]|(.6){}
			\ar@[gray]@/_/[rrrruuuu]|(.6){}
			\ar@[gray]@/_/[rrrr]_{}
			\ar@[gray]@/_/[uuuu]_{}
			\ar@[gray]@/_/[dddd]_{}
			\ar@[gray]@(ld,lu)^{}
			& & & & \bullet
			\ar@[gray]@/_/[rrrrdddd]_{}
			\ar@[gray]@/_/[uuuullll]_{}
			\ar@[gray]@/_/[ddddllll]_{}
			\ar@[gray]@/_/[rrrruuuu]|(.6){}
			\ar@[gray]@/_/[llll]_{}
			\ar@[gray]@/_/[rrrr]_{}
			\ar@[gray]@/_/[uuuu]_{}
			\ar@[gray]@/_/[dddd]_{}
			\ar@[gray]@`{+(9,4),+(-2,9)}_(.56){}
			& & & & \bullet
			\ar@/_/[rrrrdddd]|{\mu_3}
			\ar@/_/[uuuullll]|{\lambda_1}
			\ar@/_/[llll]|{-\overline{\eta_1}}
			\ar@/_/[rrrr]|{\eta_2}
			\ar@/_/[uuuu]|{-\eta_1}
			\ar@/_/[dddd]|{\overline{\eta_2}}
			\ar@`{+(9,4),+(-2,9)}|(.56){\sigma_2}
			& & & & \bullet
			\ar@[gray]@/_/[uuuullll]_{}
			\ar@[gray]@/_/[lllldddd]|(.6){}
			\ar@[gray]@/_/[llll]_{}
			\ar@[gray]@/_/[uuuu]_{}
			\ar@[gray]@/_/[dddd]_{}
			\ar@[gray]@`{+(9,4),+(-2,9)}_(.56){}
			& \cdots
			\\ \\ \\ \\
			**[d] \underset{\underset{\vdots}{ }}{\bullet} 
			\ar@[gray]@/_/[rrrruuuu]^{}
			\ar@[gray]@/_/[rrrr]_{}
			\ar@[gray]@/_/[uuuu]_{}
			\ar@[gray]@(ld,lu)^{}
			& & & & **[d] \underset{\underset{\vdots}{ }}{\bullet}  
			\ar@[gray]@/_/[uuuullll]_{}
			\ar@[gray]@/_/[rrrruuuu]|(.6){}
			\ar@[gray]@/_/[llll]_{}
			\ar@[gray]@/_/[rrrr]_{}
			\ar@[gray]@/_/[uuuu]_{}
			\ar@[gray]@`{+(9,4),+(-2,9)}_(.56){}
			& & & & **[d] \underset{\underset{\vdots}{ }}{\bullet}  
			\ar@[gray]@/_/[uuuullll]_{}
			\ar@[gray]@/_/[llll]_{}
			\ar@[gray]@/_/[rrrr]_{}
			\ar@[gray]@/_/[uuuu]_{}
			\ar@[gray]@`{+(9,4),+(-2,9)}_(.56){}
			& & & & **[d] \underset{\underset{\vdots}{ }}{\bullet}  
			\ar@/_/[uuuullll]|{\lambda_2}
			\ar@/_/[llll]|{-\overline{\eta_2}}
			\ar@/_/[uuuu]|{-\eta_2}
			\ar@`{+(9,4),+(-2,9)}|(.56){\sigma_3}
			& **[d] \underset{\underset{\ddots}{ }}{\cdots}  
		}
	\end{xy}
	\]
	\caption{Transition rates for a QBDC. 
	An arrow from the point $(m,m')$ to $(n,n')$ on the lattice 
	$\N_0 \times \N_0$ indicates that the transition rate 
	$\Tr(e_{n',n}T(e_{m,m'}))$ for a QBDC $T$ not necessarily vanishes.}
	\label{fig:transition-rates}
\end{figure}
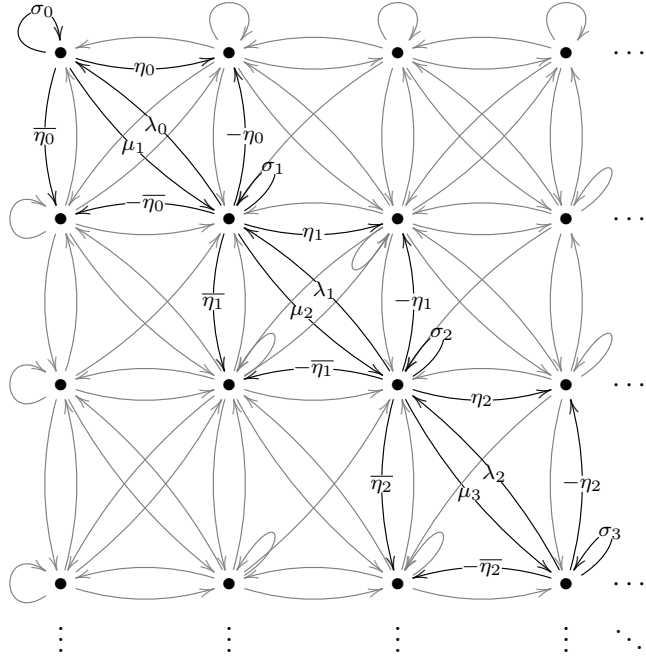

\subsection{Notational preliminaries}

The set of natural numbers is taken as $\N := \{1,2,3,\ldots\}$. 
Scalar products $\skapro{.,.}$ are linear in the first slot 
and anti-linear in the second. 
If $\H$ is a Hilbert space, then 
$\B(\H)$ denotes the algebra of bounded operators on $\H$. 
A unital completely positive linear map $T: \BH \to \BH$ 
will be referred to as a \emph{ucp-map}. 
The canonical basis of $\ell^2(\N_0)$ is denoted 
by $(e_n)_{n\in\N_0}$. 
For $n,m \in \N_0$, the \emph{matrix unit} 
$e_k \mapsto \delta_{m,k} e_n$ is denoted by $e_{n,m}$, 
and the orthogonal projection onto 
$\spann\{e_k \,|\, n \leq k \leq m\} \subset \ell^2(\N_0)$ 
by $p_{[n,m]}$. 

\begin{defn}
A \emph{quantum birth and death chain (on the non-negative integers; QBDC)} 
is a ucp-map $T: \Bl \to \Bl$ satisfying 
\begin{align}\label{eq:qbdc-def}
p_{[n+1,m-1]} \leq T(p_{[n,m]}) \leq p_{[n-1,m+1]}
\end{align}
for all $n,m \in \N_0$ with $n\leq m$. Here, $p_{[-1,m+1]}$ is 
understood to be $p_{[0,m+1]}$. 
\end{defn}

For a quantum birth and death chain $T$ 
we introduce the \emph{transition rates} 
\begin{gather}
\sigma_n := \Tr(e_{n,n}T(e_{n,n})) \ , \qquad
\mu_n := \Tr(e_{n,n}T(e_{n-1,n-1})) \ , \nonumber\\
\lambda_n := \Tr(e_{n,n}T(e_{n+1,n+1})) \ , \qquad
\eta_n := \Tr(e_{n+1,n}T(e_{n,n})) \ . \label{eq:parameters-general}
\end{gather}
The transition rates $\Tr(e_{n',n}T(e_{m,m'}))$ vanish due 
to \eqref{eq:qbdc-def} whenever $|n-m|>1$ or $|n'-m'|>1$. 
Unitality of $T$ implies that 
$\Tr(e_{n,n}T(e_{n-1,n+1})) = 0 = \Tr(e_{n,n}T(e_{n+1,n-1}))$, 
as well as the relations 
$1 = \sigma_n + \lambda_n + \mu_n$ and 
$\eta_n = -\Tr(e_{n+1,n}T(e_{n+1,n+1}))$; 
positivity gives 
$\overline{\eta_n} = \Tr(e_{n,n+1}T(e_{n,n})) 
	= -\Tr(e_{n+1,n}T(e_{n+1,n+1}))$. 
If we draw an arrow from $(m,m')$ to $(n,n')$ on the lattice 
$\N_0 \times \N_0$ and label it with the transition rate 
$\Tr(e_{n',n}T(e_{m,m'}))$ whenever it does not vanish, then 
a generic QBDC $T$ may be visualised by a diagram as in 
Figure \ref{fig:transition-rates}. 

For an unbounded operator $X$ with domain 
$\Dom(X) = D := \spann\{e_k\,|\, k\in\N_0\}$ 
let $X_{n,m} := \langle Xe_m,e_n \rangle$. 
If $T$ is a QBDC, then 
\begin{align}\label{eq:TX-for-X-unbounded}
T(X) := \text{strong-lim}_{N\in\N_0} \sum_{0\leq n,m\leq N} X_{n,m}\cdot T(e_{n,m})
\end{align}
defines another operator with domain $D$.

\subsection{Existence of normal invariant states}

We investigate the existence of normal invariant states 
using similar techniques as applied in \cite{FaR01} to the case of a 
continuous-time-semigroup with a form-generator. The case of a semigroup 
in discrete time as considered here, i.e. of powers of 
a single ucp-map, is much less technically involved, of course. 

\begin{lemma}[{cf.\ \cite[Prop.\,II.1]{FaR01}}]
\label{lem:clusters-are-inv}
	Let $\varphi$ be a normal state on a Hilbert space $\H$ and $T$ a 
	ucp-map on $\B(\H)$. 
	Then all weak cluster points of the sequence 
	\begin{align}
		\frac{1}{n} \sum_{k=1}^n \varphi \circ T^k
	\end{align}
	in the predual of $\B(\H)$ are normal invariant states for $T$.
\end{lemma}

\begin{defn}
	A sequence $(\varphi_n)_{n \in \N_0}$ of normal states 
	on $\B(\H)$ is called \emph{tight}, if for each $\varepsilon > 0$ there 
	exist a finite rank projection $p \in \B(\H)$ and $N \in \N_0$, such that
	\begin{align}
		\varphi_n(p) > 1 - \varepsilon
	\end{align}
	for all $n > N$.
\end{defn}

\begin{lemma}[{\cite[Thm.\,II.1]{FaR01}; see also \cite[Lem.\,2.2.4]{Haag06}}]
\label{lem:tight-has-cluster}
	Let $(\varphi_n)_{n \in \N_0}$ be a tight sequence of normal states 
	on $\B(\H)$. Then $(\varphi_n)_{n \in \N_0}$ possesses a weak cluster point 
	in the predual of $\B(\H)$. 
\end{lemma}

Given a self-adjoint operator $Y$ and a measurable subset $A \subseteq \R$, 
denote $Y$'s spectral projection corresponding to $A$ by $\chi_A(Y)$. 
For a one-parameter-semigroup version of the 
following lemma -- without the fall-off statement, 
see \cite[Thm.\,II.1]{FaR01}.

\begin{lemma}\label{lem:far_ex_inv}
	Let $T$ be a QBDC. 
	Let $X,Y$ be self-adjoint operators on $\ell^2(\N_0)$ with 
	$D \subseteq \Dom(X)$, $\Dom(Y)$, 
	and assume that $X$ is positive, $Y$ is bounded from below 
	by $-b$ ($b>0$) and $Y$'s spectral projections associated to 
	bounded sets are finite dimensional. 
	If
	\begin{align}\label{eq:lem-far-ex-inv}
		\sum_{k=1}^n \langle T^k(Y) \xi,\xi \rangle \leq \langle X\xi, \xi \rangle
	\end{align}
	for all $n \in \N$ and $\xi \in D$, then $T$ has a normal invariant state 
	$\varphi$ such that 
	$\varphi(\chi_{[c,d]}(Y)) \leq \frac{b}{c}$ 
	holds for all $0 < c \leq d$.
\end{lemma}
\begin{proof}
	Observe that for any (unbounded) operator $A$ with 
	domain $D \subseteq \Dom(A)$ which satisfies $\langle A \xi, \xi \rangle \geq 0$ 
	for all $\xi \in D$, one has $\langle T(A) \xi, \xi \rangle \geq 0$ 
	for all $\xi \in D$: 
	Namely, with $M\geq 0$ large enough such that $\xi = p_{[0,M]} \xi$,  
	\begin{align} 
		\langle T(A) \xi, \xi \rangle 
		= \langle p_{[0,M]} T(A) p_{[0,M]} \xi, \xi \rangle 
		\stackrel{\eqref{eq:qbdc-def} \& \eqref{eq:TX-for-X-unbounded}}= 
			\langle p_{[0,M]} T(p_{[0,M+1]} A p_{[0,M+1]}) p_{[0,M]} \xi, \xi \rangle 
		\geq 0 \ .
	\end{align} 
	Iterating this, one obtains $\langle T^n(A) \xi, \xi \rangle \geq 0$ 
	for all $\xi \in D$ and all $n \in \N$. 
	
	Now we have for each $r > 0$ 
	\begin{align}
		Y \geq -b \chi_{(-\infty,r]}(Y) + r \chi_{(r,\infty)}(Y) = -(b + r) \chi_{(-\infty,r]}(Y) + r \one \ .
	\end{align}
	So, 
	\begin{align}
		-(b + r) \sum_{k=1}^n \langle T^k(\chi_{(-\infty,r]}(Y)) \xi,\xi \rangle + nr \|\xi\|^2 &\leq \langle X\xi, \xi \rangle \ . 
	\end{align}
	Bringing $nr \|\xi\|^2$ to the other side and dividing by 
	$-(b + r)n$ gives  
	\begin{align}
		\frac1n \sum_{k=1}^n \langle T^k(\chi_{(-\infty,r]}(Y)) \xi,\xi \rangle &\geq \frac{r \|\xi\|^2}{b + r} - \frac{\langle X\xi, \xi \rangle}{n(b + r)} \ .
	\end{align}
	Choosing $\xi$ with $\|\xi\| = 1$, 
	the last line says that the sequence of states $\frac1n \sum_{k=1}^n T_*^k(|\xi \rangle\langle \xi|)$ 
	is tight, since $\chi_{(-\infty,r]}(Y) = \chi_{[-b,r]}(Y)$ 
	is a finite rank projection 
	(here, $T_*$ denotes the predual of the map $T: \Bl \to \Bl$).
	
	Let $\varphi$ be a weak cluster point of that sequence 
	according to Lemma \ref{lem:tight-has-cluster}. 
	By Lemma \ref{lem:clusters-are-inv}, $\varphi$ is a normal 
	invariant state. 
	From $Y \geq -b\one + c\chi_{[c,d]}(Y)$ and 
	\eqref{eq:lem-far-ex-inv} we get 
	\begin{align}
	-nb + c\sum_{k=1}^n \skapro{T^k(\chi_{[c,d]}(Y))\xi,\,\xi} \leq \skapro{X\xi,\,\xi} \ ,
	\end{align}
	which shows the estimate $\varphi(\chi_{[c,d]}(Y)) \leq \frac{b}{c}$.
\end{proof}

See \cite[Thm.\,IV.1]{FaR01} for a one-parameter-semigroup version of the following proposition.

\begin{prop}\label{prop:far_ex_inv}
	Let $T$ be a QBDC. 
	Let $X,Y$ be self-adjoint (unbounded) operators on $\ltwo$ such that $D \subseteq \Dom(X), \Dom(Y)$, $X$ is positive and 
	$Y$ is bounded from below. Assume that $Y$'s spectral projections associated to 
	bounded sets are finite dimensional. If
	\begin{align}
		\langle (T(X) - X)\xi,\xi\rangle \leq - \langle Y\xi, \xi \rangle
	\end{align}
	for all $\xi \in D$, then $T$ has a normal invariant state
	$\varphi$ such that 
	$\varphi(\chi_{[c,d]}(Y)) \leq \frac{b}{c}$ 
	holds for all $0 < c \leq d$ and some $b>0$.
\end{prop}
\begin{proof}
	We have, for all $\xi \in D$, 
	\begin{align}
	\skapro{X\xi,\,\xi} \geq \skapro{(X - T^{n+1}(X))\xi,\,\xi} 
	= \sum_{k=0}^n \skapro{T^k(X - T(X))\xi,\,\xi} 
	\geq \sum_{k=0}^n \skapro{T^k(Y)\xi,\,\xi} \ .
	\end{align}
	Now the claim follows on using 
	Lemma \ref{lem:far_ex_inv}.
\end{proof}

We come to the main result.
If $\mu_k \neq 0$ for all $k$, then put 
$\pi_n := \frac{\lambda_0\lambda_1\cdots\lambda_{n-1}}{\mu_1\mu_2\cdots\mu_{n}}$. 
A normal state $\varphi$ on $\Bl$ is said to be of \emph{exponential fall-off} 
if there are constants $C,\gamma>0$ such that 
$\varphi(e_{n,n}) \leq Ce^{-\gamma n}$ for all $n\in\N_0$. 

\begin{thm}\label{thm:main-result}
	Let $T$ be a QBDC and let $\mu_n,\lambda_n$ be as in 
	\eqref{eq:parameters-general}. 
	\begin{enumerate}
	\item Suppose that $\lambda_k,\mu_k\neq 0$ for all $k$ and 
	that $\kappa := {\lim\inf}_n \frac 1n \ln(\pi_n^{-1}) > 0$. 
	If 
\begin{align}\label{eq:main-ineq}
	\lim\inf_n \frac{\lambda_n\mu_{n+1}}{4|\eta_n|^2} 
		> \frac{e^{-\kappa}}{(1 - e^{-\kappa})^2} \ ,
\end{align}
	then $T$ has a normal invariant state of exponential fall-off. 
	\item Suppose that $\lambda_n \neq 0$ for all $n \in \N$ 
	and that there is $m \in \N$ such that $\mu_n < \lambda_n$ for all $n > m$. 
	If 
\begin{align}\label{eq:main-ineq-neg}
	\lim\inf_n \frac{(\lambda_n-\mu_n)(\lambda_{n+1}-\mu_{n+1})}{4|\eta_n|^2} 
		> 1 \ , 
\end{align}
	then $T$ has no normal invariant state. 
	\end{enumerate}
\end{thm}
\begin{proof}
\emph{Part 1:} 
	Choose $e^{-\kappa} < t < 1$ and $0 < r < 1 - t$ such that still 
	\begin{align}\label{eq:main-ineq-2}
	\frac{\lambda_n\mu_{n+1}}{4|\eta_n|^2} 
		> \frac{t}{(1 - r - t)^2}
	\end{align}
	holds for all $n > N$, with some $N \in \N$. 
	Let $X := \diag(x_0, x_1, x_2, \ldots)$ and $Y := \diag(y_0, y_1, y_2, \ldots)$, 
	where $x_n := \sum_{k=1}^n (\mu_k\pi_k)^{-1} t^k$ for $n \in \N_0$ and 
	$y_n := r\pi_n^{-1} t^n$ for $n > N$ (the $y_n$'s for $n=1,2,\ldots,N$ will 
	be chosen later); 
	the domains of $X$ and $Y$ are taken as 
	$D(X) = D(Y) = D := \text{span}\{e_n: \; n \in \N_0\} \subset \ell^2(\N_0)$. 
	Since $y_n = r\cdot(\pi_n^{-1} e^{-n\kappa})\cdot(t e^\kappa)^n \rightarrow \infty$, 
	$Y$'s spectral projections associated to bounded 
	sets are finite dimensional. 
	Both $X$ and $Y$ are densely defined and semi-bounded symmetric operators, hence, 
	they possess self-adjoint extensions. 
	We have
\begin{align}
	(T(X))_{n,n} &= \sigma_n x_n + \lambda_n x_{n+1} + \mu_n x_{n-1} 
\end{align}
	for $n\geq1$, $(T(X))_{0,0} = \lambda_0 x_{1} = \lambda_0 t$, and
\begin{gather}
	(T(X))_{n,1+n} = \eta_n (x_{n} - x_{n+1}) \ , \qquad
	(T(X))_{1+n,n} = \overline{(T(X))_{n,1+n}} \ , \nonumber\\
	(T(X))_{n,k+n} = (T_\psi(X))_{k+n,n} = 0 \ , 
\end{gather}
	for all $n\geq 0$ and $k > 1$, i.e. $T(X)$ is a tridiagonal operator.
	Using $1 = \sigma_n + \lambda_n + \mu_n$, we obtain
\begin{align}
	(T(X) - X)_{n,n} &= \lambda_n (x_{n+1} - x_n) - \mu_n (x_n - x_{n-1}) 
		= \frac{\lambda_n}{\mu_{n+1}\pi_{n+1}} t^{n+1} - \pi_n^{-1} t^n \\
		&= \pi_n^{-1} (t - 1) t^n \nonumber
\end{align}
	for all $n\geq1$. 
	We want to show that $T(X) - X \leq -Y$, i.e. $0 \leq X - T(X) - Y$, because then the 
	existence of a normal invariant state follows from Proposition \ref{prop:far_ex_inv}. 
	To this end, by \cite[Prop.\,1]{Bar78} and 
	since $X - Y - T(X)$ is tridiagonal and symmetric, it suffices to show that 
	\begin{enumerate}[1)]
	\item 
	the diagonal entries of $X - T(X) - Y$ are positive, 
	\item 
	the following expressions are positive: 
	\begin{align} 
		&(X - T(X) - Y)_{n,n}(X - T(X) - Y)_{n+1,n+1} \\
			&\quad - 4(X - T(X) - Y)_{n+1,n}(X - T(X) - Y)_{n,n+1} \ . \nonumber
	\end{align}
	\end{enumerate}
	Firstly,
\begin{align}
	(X - T(X) - Y)_{n,n} &= \pi_n^{-1} (1 - t) t^n - \pi_n^{-1} r t^n 
		= \pi_n^{-1} (1 - r - t) t^n > 0 
\end{align}
	for all $n > N$. Secondly, for all $n > N$ we have
\begin{align}
		&(X - T(X) - Y)_{n,n}(X - T(X) - Y)_{n+1,n+1} \\
			&\qquad - 4(X - T(X) - Y)_{n+1,n}(X - T(X) - Y)_{n,n+1} \nonumber\\
			&\quad = \pi_n^{-1}(1 - r - t) t^n 
				\cdot \pi_{n+1}^{-1}(1 - r - t) t^{n+1} 
			- 4|\eta_n|^2 \frac{t^{2n+2}}{\mu_{n+1}^2\pi_{n+1}^2} \nonumber\\
			&\quad = \frac{1}{\mu_{n+1}^2\pi_{n+1}^2} 
				\left(\lambda_{n}\mu_{n+1} (1 - r - t)^2 - 4|\eta_n|^2 t\right) t^{2n+1} 
			\stackrel{\eqref{eq:main-ineq-2}}> 0 \ . \nonumber
\end{align}
	
	Finally, choosing $y_n \in \R$ sufficiently small ($y_n \ll 0$) 
	for $n = 1,2,\ldots,N$, both conditions 1) and 2) can be 
	fulfilled for all $n \in \N_0$. 
	
	As $Y$ is diagonal with exponentially growing eigenvalues 
	$y_n$, the fall-off statement in Proposition 
	\ref{prop:far_ex_inv} shows 
	that $T$ has a normal invariant state of exponential fall-off. 
	
	\medskip
	
\emph{Part 2:} 
	First, from the assumption $\lambda_n \neq 0$ it follows 
	that for any invariant normal state $\varphi$ of $T$ 
	there is $m \in \N$ such that 
	$\varphi(e_{n,n}) \neq 0$ for all $n \geq m$: 
	For, if $\varphi(e_{n+1,n+1}) = 0$, 
	then, by positivity of $\varphi$ we also have 
	$\varphi(e_{n,n+1}) = 0 = \varphi(e_{n+1,n})$. 
	Hence, if $\varphi(e_{n,n}) \neq 0$ but $\varphi(e_{n+1,n+1}) = 0$, 
	then 
	$\varphi(T(e_{n+1,n+1})) = \lambda_n\varphi(e_{n,n}) 
		+ \mu_n\varphi(e_{n+2,n+2}) > 0$, 
	in contradiction to $\varphi$ being invariant. 
	
	Now it suffices to find a sequence $z_0,z_1,z_2,\ldots$ 
	of real numbers, a sequence $\varepsilon_0,\varepsilon_1,\ldots \geq 0$ 
	and a constant $C>0$ with the following properties: 
	\begin{enumerate}[1)]
	\item $\varepsilon_n \neq 0$ for infinitely many $n\in\N$, 
	\item $|z_{n+1}-z_n| \leq C$ for all $n \in \N$, 
	\item $\skapro{(T(Z)-Z-\varepsilon)\xi,\,\xi} \geq 0$ 
	for all $\xi \in D$, where $Z$ denotes the (unbounded) operator 
	$Z := \diag(z_0,z_1,z_2,\ldots)$, and 
	$\varepsilon := \diag(\varepsilon_0,\varepsilon_1,\varepsilon_2,\ldots)$. 
	\end{enumerate}
	{\sloppy
	Namely, if $\varphi$ is a normal invariant state on $\Bl$, 
	choose $n,m \in \N$, $m > n$, with 
	$\varphi(p_{[0,m]}^\bot) < \left(\frac{\delta}{18C}\right)^2$ 
	for $\delta := \varepsilon_n\varphi(e_{n,n}) > 0$. 
	With $Z^{\wedge m} := \diag(z_0,z_1,\ldots,z_{m-1},z_m,z_m,z_m,\ldots) \in \Bl$ 
	we have $\|T(Z^{\wedge m}) - Z^{\wedge m}\| \leq 3C$, 
	since $T(Z^{\wedge m}) - Z^{\wedge m}$ is tridiagonal 
	with entries bounded by $C$ 
	(note that $0 \leq \lambda_k,\mu_k,|\eta_k|\leq 1$). 
	Hence, 
	\begin{align}
	&\varphi(T(Z^{\wedge m+1}) - Z^{\wedge m+1}) 
	= \varphi(\underbrace{p_{[0,m]}(T(Z^{\wedge m+1}) - Z^{\wedge m+1})p_{[0,m]}}_{\geq \varepsilon_n e_{n,n} \ \ \text{, by 3)}}) \nonumber\\
	&\hspace{2.5em}+ \underbrace{\varphi(p_{[0,m]}^\bot(T(Z^{\wedge m+1}) - Z^{\wedge m+1})p_{[0,m]})}_{\leq \sqrt{\varphi(p_{[0,m]}^\bot)}\|T(Z^{\wedge m+1}) - Z^{\wedge m+1}\|} 
	+ \underbrace{\varphi(p_{[0,m]}(T(Z^{\wedge m+1}) - Z^{\wedge m+1})p_{[0,m]}^\bot)}_{\leq \sqrt{\varphi(p_{[0,m]}^\bot)}\|T(Z^{\wedge m+1}) - Z^{\wedge m+1}\|} \nonumber\\
	&\hspace{2.5em}+ \underbrace{\varphi(p_{[0,m]}^\bot(T(Z^{\wedge m+1}) - Z^{\wedge m+1})p_{[0,m]}^\bot)}_{\leq \sqrt{\varphi(p_{[0,m]}^\bot)}\|T(Z^{\wedge m+1}) - Z^{\wedge m+1}\|} 
	> \delta - 3\cdot \frac{\delta}{18C}\cdot 3C > 0 \ ,
	\end{align}
	such that $\varphi$ cannot be invariant. 
	}
	
	Let us turn to the choice of $Z$. 
	First, let $N \in \N_0$ be large enough such that 
	for all $n>N$ there are $\varepsilon_n>0$ with $\lambda_n-\mu_n-\varepsilon_n > 0$ and 
	\begin{align}\label{eq:choice-of-N-and-eps}
	\frac{(\lambda_n-\mu_n-\varepsilon_n)(\lambda_{n+1}-\mu_{n+1}-\varepsilon_{n+1})}{4|\eta_n|^2} 
		> 1 \ .
	\end{align}
	Then define 
	\begin{align}
	Z := \diag(\underbrace{0,0,\ldots,0}_{N+2\text{ times}},1,2,3,\ldots) \quad 
	\text{and}\quad 
	\varepsilon := \diag(\underbrace{0,0,\ldots,0}_{N+1\text{ times}},\varepsilon_{N},\varepsilon_{N+1},\ldots) \ .
	\end{align}
	As we have, for $n>N$, 
\begin{align}
	(T(Z) - Z - \varepsilon)_{n,n} &= \lambda_n(z_{n+1}-z_n) 
		- \mu_n(z_n - z_{n-1}) - \varepsilon_n = \lambda_n - \mu_n - \varepsilon_n \ , 
		\nonumber\\
	(T(Z) - Z - \varepsilon)_{n,1+n} &= \eta_n(z_{n+1}-z_n) = \eta_n 
		= \overline{(T(Z) - Z - \varepsilon)_{1+n,n}} \ , 
\end{align}
	and $(T(Z) - Z - \varepsilon)_{n,m} = 0$ 
	otherwise, we see that $T(Z) - Z \geq \varepsilon$. 
\end{proof}

\subsection{Example: Generalised one-atom masers}
\label{sec:generalised-maser}

Let $\psi$ be a state on $M_2$, parametrised 
by $0 \leq \lambda \leq 1$ and 
$\zeta \in \mathbb D = \{z \in \C \,|\, 
	|z| \leq 1\}$ 
via 
\begin{align}
\psi(x) = \Tr \left(\begin{pmatrix}
\lambda & i\bar\nu \\ -i\nu & 1-\lambda
\end{pmatrix} x \right) \ , \qquad x \in M_2 \ , \qquad 
\text{where } \nu := i\sqrt{\lambda(1-\lambda)}\zeta \ .
\end{align}
Let $a = \diag(1,\alpha_1,\alpha_2,\ldots)$ 
and $b = \diag(0,\beta_1,\beta_2,\ldots)$
be infinite diagonal matrices with 
$-1 \leq \alpha_n,\beta_n \leq 1$ ($n \in \N$) 
and $a^2 + b^2 = \one$. 
So $a,b$ give rise to bounded operators 
on $\ell^2(\N_0)$. 
With $s \in \Bl$ denoting the right-shift $e_n \mapsto e_{n+1}$, 
let 
\begin{align}\label{eq:formula-for-T-psi}
T_\psi(x) &:= \lambda(s^*as x s^*as + s^*b x bs) 
	+ (1-\lambda)(bsxs^*b + axa) \nonumber\\
	&\quad- \bar\nu(axbs - bsxs^*as) + \nu(s^*as x s^*b - s^*b x a) \ , 
	\qquad x \in \Bl \ ,
\end{align}
be the transition operator on $\Bl$ associated with $\psi$, 
cf. \cite{BGKRSS13}. 
Then $T_\psi$ defines a QBDC, which can be regarded as 
a generalisation of the Jaynes-Cummings one-atom-maser. 
The transition rates of $T_\psi$ are shown in Figure \ref{fig:action}.

\begin{figure}[bht]
	\[
	\begin{xy}
		\xymatrix@C=1.8pc@R=1.8pc{
			\bullet
			\ar@/_/[rrrrdddd]|(.6){(1\hspace{-.1ex}-\hspace{-.1ex}\lambda)\beta_1^2\hspace{4ex}}
			\ar@/_/[rrrr]_{\alpha_1\beta_1\nu}
			\ar@/_/[dddd]_{\alpha_1\beta_1\bar\nu}
			\ar@(l,u)^(.7){(\lambda\alpha_1^2+(1\hspace{-.1ex}-\hspace{-.1ex}\lambda)\alpha_0^2)}
			& & & & \bullet
			\ar@/_/[rrrrdddd]|(.6){(1\hspace{-.1ex}-\hspace{-.1ex}\lambda)\beta_1\beta_2\hspace{6ex}}
			\ar@/_/[llll]_{-\alpha_0\beta_1\bar\nu}
			\ar@/_/[rrrr]_{\alpha_1\beta_2\nu}
			\ar@/_/[dddd]_{\alpha_2\beta_1\bar\nu}
			\ar@(lu,ru)^{(\lambda\alpha_1\alpha_2+(1\hspace{-.1ex}-\hspace{-.1ex}\lambda)\alpha_0\alpha_1)}
			& & & & \bullet
			\ar@/_/[llll]_{-\alpha_0\beta_2\bar\nu}
			\ar@/_/[dddd]_{\alpha_3\beta_1\bar\nu}
			\ar@(lu,ru)^{(\lambda\alpha_1\alpha_3+(1\hspace{-.1ex}-\hspace{-.1ex}\lambda)\alpha_0\alpha_2)}
			& \hspace{-1em}\cdots
			\\ \\ \\ \\
			\bullet
			\ar@/_/[rrrrdddd]|(.6){(1\hspace{-.1ex}-\hspace{-.1ex}\lambda)\beta_1\beta_2\hspace{6ex}}
			\ar@/_/[rrrr]_{\alpha_2\beta_1\nu}
			\ar@/_/[uuuu]_{\hspace{-0.5ex}-\alpha_0\beta_1\nu}
			\ar@/_/[dddd]_{\alpha_1\beta_2\bar\nu}
			\ar@`{+(9,4),+(-2,9)}_(.56){\text{\begin{tiny}$(\lambda\alpha_1\alpha_2\!\!+\!\!(1\!\!-\!\!\lambda)\alpha_0\alpha_1)$\end{tiny}}}
			& & & & \bullet
			\ar@/_/[rrrrdddd]|(.6){(1\hspace{-.1ex}-\hspace{-.1ex}\lambda)\beta_2^2\hspace{4ex}}
			\ar@/_/[uuuullll]_{\hspace{-1ex}\lambda\beta_1^2}
			\ar@/_/[llll]_{-\alpha_1\beta_1\bar\nu}
			\ar@/_/[rrrr]_{\alpha_2\beta_2\nu}
			\ar@/_/[uuuu]_{\hspace{-0.5ex}-\alpha_1\beta_1\nu}
			\ar@/_/[dddd]_{\alpha_2\beta_2\bar\nu}
			\ar@`{+(9,4),+(-2,9)}_(.56){\text{\begin{tiny}$(\lambda\alpha_2^2\!\!+\!\!(1\!\!-\!\!\lambda)\alpha_1^2)$\end{tiny}}}
			& & & & \bullet
			\ar@/_/[uuuullll]_{\hspace{-1ex}\lambda\beta_1\beta_2}
			\ar@/_/[llll]_{-\alpha_1\beta_2\bar\nu}
			\ar@/_/[uuuu]_{\hspace{-0.5ex}-\alpha_2\beta_1\nu}
			\ar@/_/[dddd]_{\alpha_3\beta_2\bar\nu}
			\ar@`{+(9,4),+(-2,9)}_(.56){\text{\begin{tiny}$(\lambda\alpha_2\alpha_3\!\!+\!\!(1\!\!-\!\!\lambda)\alpha_1\alpha_2)$\end{tiny}}}
			& \hspace{-1em}\cdots
			\\ \\ \\ \\
			**[d] \underset{\underset{\vdots}{ }}{\bullet} 
			\ar@/_/[rrrr]_{\alpha_3\beta_1\nu}
			\ar@/_/[uuuu]_{\hspace{-0.5ex}-\alpha_0\beta_2\nu}
			\ar@`{+(9,4),+(-2,9)}_(.56){\text{\begin{tiny}$(\lambda\alpha_1\alpha_3\!\!+\!\!(1\!\!-\!\!\lambda)\alpha_0\alpha_2)$\end{tiny}}}
			& & & & **[d] \underset{\underset{\vdots}{ }}{\bullet}  
			\ar@/_/[uuuullll]_{\hspace{-1ex}\lambda\beta_1\beta_2}
			\ar@/_/[llll]_{-\alpha_2\beta_1\bar\nu}
			\ar@/_/[rrrr]_{\alpha_3\beta_2\nu}
			\ar@/_/[uuuu]_{\hspace{-0.5ex}-\alpha_1\beta_2\nu}
			\ar@`{+(9,4),+(-2,9)}_(.56){\text{\begin{tiny}$(\lambda\alpha_2\alpha_3\!\!+\!\!(1\!\!-\!\!\lambda)\alpha_1\alpha_2)$\end{tiny}}}
			& & & & **[d] \underset{\underset{\vdots}{ }}{\bullet}  
			\ar@/_/[uuuullll]_{\hspace{-1ex}\lambda\beta_2^2}
			\ar@/_/[llll]_{-\alpha_2\beta_2\bar\nu}
			\ar@/_/[uuuu]_{\hspace{-0.5ex}-\alpha_2\beta_2\nu}
			\ar@`{+(9,4),+(-2,9)}_(.56){\text{\begin{tiny}$(\lambda\alpha_3^2\!\!+\!\!(1\!\!-\!\!\lambda)\alpha_2^2)$\end{tiny}}}
			& **[d] \hspace{-1em}\underset{\underset{\ddots}{ }}{\cdots}  
		}
	\end{xy}
	\]
	\caption{Action of $T_\psi$, where $i\zeta\sqrt{\lambda(1-\lambda)}$ is abbreviated by $\nu$; cf. \cite[Fig.\,III.1]{BGKRSS13}.}
	\label{fig:action}
\end{figure}

For this class of QBDCs 
one obtains, as a consequence of Theorem~\ref{thm:main-result}: 

\begin{prop}\label{cor:main-result}
	Let $\hat\beta := \lim\sup_n |\beta_n| > 0$, 
	$\check\beta := \lim\inf_n |\beta_n| > 0$ 
	and put $\hat\alpha := \sqrt{1 - \hat\beta^2}$, 
	$\check\alpha := \sqrt{1 - \check\beta^2}$. 
	We have: 
	\begin{enumerate}
	\item If $\lambda < \frac 12 - \frac{\hat\alpha}{\hat\beta}|\nu|$, then 
		$T_\psi$ has a normal invariant state of exponential fall-off.
	\item If $\lambda > \frac 12 + \frac{\check\alpha}{\check\beta}|\nu|$, then 
		$T_\psi$ admits no normal invariant state.
	\end{enumerate}
\end{prop}
\begin{proof}
	For $T_\psi$, the coefficients $\lambda_n,\mu_n,\eta_n$ 
	are given by 
	$\lambda_n = \lambda\beta_{n+1}^2$, 
	$\mu_n = (1-\lambda)\beta_n^2$, 
	$\eta_n = -\alpha_n\beta_{n+1}\bar\nu$, 
	as one may read off from Figure \ref{fig:action} 
	or \eqref{eq:formula-for-T-psi}. 
	Therefore, 
	$\kappa = \lim\inf\ln\frac{\mu_n}{\lambda_{n-1}} 
		= \ln\frac{1-\lambda}\lambda$, 
	which is bigger than $0$ if and only if $\lambda < \frac 12$. 
	Now the right-hand side of \eqref{eq:main-ineq} 
	reads 
	$\frac{e^{-\kappa}}{(1 - e^{-\kappa})^2} 
	= \frac{\frac\lambda{1-\lambda}}{\left(1-\frac\lambda{1-\lambda}\right)^2} 
	= \frac{\lambda(1-\lambda)}{(1-2\lambda)^2}$. 
	For the left-hand side we obtain 
	$\frac{\lambda_n\mu_{n+1}}{4|\eta_n|^2} 
	= \frac{\lambda\beta_{n+1}^2(1-\lambda)\beta_n^2}{4\alpha_n^2\beta_{n+1}^2|\nu|^2} 
	= \frac{\lambda(1-\lambda)\beta_n^2}{4|\nu|^2\alpha_n^2}$, 
	which can be estimated from below by 
	$\frac{\lambda(1-\lambda)\beta^2}{4|\nu|^2\alpha^2}$, 
	for $n$ sufficiently large. 
	Now the first statement easily follows. 
	The second statement is obtained from part 2 of 
	Theorem \ref{thm:main-result} by a similarly easy 
	calculation. 
\end{proof}

It is not known to the author how sharp condition \eqref{eq:main-ineq} is. 
To address this question 
we consider the special case where $\alpha_n = \alpha$, $\beta_n = \beta$ 
are constant. This is referred to as \enquote{toy model} 
in \cite{BGKRSS13}. 
Remarkably, for this toy model, for pure atomic states $\psi$ 
and for $\alpha>0$, 
the condition $\lambda < \frac 12 - |\frac\alpha\beta\nu|$ 
in Proposition \ref{cor:main-result}.1 marks the full set of parameters 
$(\lambda,\zeta)$ with $\lambda < \frac 12$ 
for which a \emph{pure} invariant state exists: 
according to \cite[Prop.\,5.1]{BGKRSS13}, a pure 
invariant state exists in this case 
if and only if $\lambda < \frac 12 (1 - \alpha)$ holds. 
With $|\zeta| = 1$ (which amounts for the atomic state $\psi$ being pure), 
one finds 
\begin{gather}
\lambda < \frac 12 - \left|\frac\alpha\beta\nu\right| 
	= \frac 12 - \left|\frac\alpha\beta\right|\lambda(1-\lambda) 
\qquad \Longleftrightarrow \qquad 
\lambda^2 - \lambda + \underbrace{\frac{|\beta|^2}4}_{= \frac{1-|\alpha|^2}4} > 0 \nonumber\\
\Longleftrightarrow \qquad 
\left(\lambda - \frac{1-|\alpha|}2\right) \left(\lambda - \frac{1+|\alpha|}2\right) > 0 \ . \label{eq:umgef-main-ineq}
\end{gather}
For $0 < \lambda < \frac 12$ and $-1 \leq \alpha \leq 1$, the second 
factor on the left-hand side in \eqref{eq:umgef-main-ineq} is negative, hence, 
\eqref{eq:umgef-main-ineq} is equivalent to $\lambda < \frac 12 (1 - |\alpha|)$. 
If $-1 < \alpha < 0$, then there exist pure states $\psi$ in the 
upper Bloch hemisphere ($\lambda > \frac 12$) for which 
$T_\psi$ admits an invariant state. These are obviously not 
captured by Theorem \ref{thm:main-result}. 

The following proposition extends the non-existence statement in 
Proposition \ref{cor:main-result} and allows to strengthen 
the above observations. 
If $\alpha > 0$, it determines a parameter region in the lower Bloch hemisphere 
where the toy model does not admit 
invariant states. 

\begin{prop}
Let $\alpha_n = \alpha$ and $\beta_n = \beta$ ($n>0$) 
with $-1 < \alpha,\beta < 1$ and $\alpha^2 + \beta^2 = 1$. 
For $0 < \lambda < 1$, if 
$\beta^2 < \frac 1{1 + \frac{(1-2\lambda)^2}{4|\nu|^2}}$ and 
$\frac{1-\alpha}{|\beta|} < \frac{|\nu|}{1-\lambda}$, then $T_\psi$ has 
no normal invariant state. 
\end{prop}
\begin{proof}
The idea is to construct, for any given normal state, an observable 
whose expectation value under this state strictly increases or 
decreases. 
These observables will be built from the \enquote{number operator} 
$N = \diag(0,1,2,3,\ldots)$, which is an unbounded operator 
on $\ell^2(\N_0)$ with domain $D = \spann\{e_n\,|\,n\in\N_0\}$, 
or rather from its bounded truncations 
$N^{\wedge m} := \diag(0,1,\ldots,m-1,m,m,m,\ldots)$ 
for $m \geq 0$ and $N^{\wedge m} := 0$ for $m < 0$. 

For sequences 
$\mb x = (x_k)_{k\in\N}$, $\mb y = (y_k)_{k\in\N} \subset \C$, 
let 
\begin{align}
A_{\mb x,\mb y} := N + \sum_{k>0}\left( \left( x_k + y_kN \right) (s^*)^k 
	+ s^k \left( \overline{x_k} + \overline{y_k}N \right) \right) \  
\end{align}
as an unbounded operator on $\ell^2(\N_0)$ with domain $D$, and 
\begin{align}
A_{\mb x,\mb y}^{\wedge m} := N^{\wedge m} + \sum_{k=1}^{2m}\left( \left( x_k + y_kN^{\wedge m-\lfloor \frac k2 \rfloor} \right) (s^*)^k 
	+ s^k \left( \overline{x_k} + \overline{y_k}N^{\wedge m-\lfloor \frac k2 \rfloor} \right) \right) \ .
\end{align}
The proof is split into several steps: 

\medskip

\emph{Step 1:} If $\beta^2 < \frac 1{1 + \frac{(1-2\lambda)^2}{4|\nu|^2}}$, 
then for each $C \in \R$ there exist sequences 
$\mb x = (x_k), \mb y = (y_k) \subset \C$ and $C' > 0$, such that 
$|y_k| < C'$ and $T_\psi(A_{\mb x,\mb y}) - A_{\mb x,\mb y} = C\cdot \one$. 

\emph{Proof:} 
From the representation of the transition rates 
in Figure \ref{fig:action} one reads 
off that for an operator $X$ on $\ell^2(\N_0)$ with domain $D$ and 
for $n > 0$, $k \geq 0$ we have 
\begin{align}
T_\psi(X)_{n,n+k} &= \alpha^2 X_{n,n+k} + (1-\lambda)\beta^2 X_{n-1,n+k-1} 
	+ \lambda\beta^2 X_{n+1,n+k+1} \\
	&\quad+ \alpha\beta\nu(X_{n,n+k-1} - X_{n+1,n+k}) 
	+ \alpha\beta\bar\nu(X_{n-1,n+k} - X_{n,n+k+1}) \ .\nonumber
\end{align}
Hence, 
\begin{align}\label{eq:T-transitions-via-coords}
&(T_\psi(X) - X)_{n,n+k} = (1-\lambda)\beta^2(X_{n-1,n+k-1} - X_{n,n+k}) 
	+ \lambda\beta^2(X_{n+1,n+k+1} - X_{n,n+k}) \nonumber\\
	&\hspace{6em}+ \alpha\beta\nu(X_{n,n+k-1} - X_{n+1,n+k}) 
	+ \alpha\beta\bar\nu(X_{n-1,n+k} - X_{n,n+k+1}) \ . 
\end{align}
Inserting $X = A_{\mb x,\mb y}$, one finds that 
$T_\psi(A_{\mb x,\mb y}) - A_{\mb x,\mb y} \stackrel!= C\cdot \one$ 
implies for $k>0$, 
\begin{gather}
0 \stackrel{!}= (1-\lambda)\beta^2(-y_k) 
	+ \lambda\beta^2 y_k 
	+ \alpha\beta\nu(-y_{k-1}) 
	+ \alpha\beta\bar\nu(-y_{k+1}) \nonumber\\
\Longrightarrow\quad y_{k+1} = \frac{2\lambda - 1}{\bar\nu}\frac{\beta}{\alpha} y_k - 
		\frac\nu{\bar\nu} y_{k-1} \ . \label{eq:rec-for-ys}
\end{gather}
The roots of the characteristic polynomial 
$x^2 - \frac{2\lambda - 1}{\bar\nu}\frac{\beta}{\alpha}x + \frac\nu{\bar\nu}$ 
of the recurrence relation for $\mb y$ are given by 
\begin{align}
x_{1/2} &= \frac 12\left( \frac{2\lambda - 1}{\bar\nu}\frac{\beta}{\alpha} 
	\pm \sqrt{\frac{(2\lambda - 1)^2}{\bar\nu^2}\frac{\beta^2}{\alpha^2} 
		- 4 \frac\nu{\bar\nu}} \right) 
		\nonumber\\ 
	&= \frac{2\lambda - 1}{2\bar\nu}\frac{\beta}{\alpha}\left( 1 
	\pm \sqrt{1 - \frac{4|\nu|^2\alpha^2}{(2\lambda - 1)^2\beta^2}} \right) \ .
\end{align}
As 
\begin{align}
\beta^2 < \frac 1{1 + \frac{(1-2\lambda)^2}{4|\nu|^2}} 
\ \Leftrightarrow \ 
\beta^2 \left( 1 + \frac{4|\nu|^2}{(2\lambda-1)^2} \right) < \frac{4|\nu|^2}{(2\lambda-1)^2} 
\ \Leftrightarrow \ 
1 > \frac{4|\nu|^2\alpha^2}{(2\lambda - 1)^2\beta^2} \ ,
\end{align}
the discriminant is negative, and therefore the absolute values of the 
roots $x_1, x_2$ are 
\begin{align}
|x_{1/2}|^2 = \left(\frac{2\lambda - 1}{2\bar\nu}\frac{\beta}{\alpha}\right)^2 \left( 1 
	+ \frac{4|\nu|\alpha^2}{(2\lambda - 1)^2\beta^2} - 1 \right) = 1 \ .
\end{align}
Now the general solution to the recurrence \eqref{eq:rec-for-ys} 
for $\mb y$, given by 
$y_k = x_1^k \cdot y' + x_2^k \cdot y''$ for some $y',y'' \in \C$, 
shows that the sequence $\mb y$ is necessarily bounded. 

Inserting $X = A_{\mb x,\mb y}$ into \eqref{eq:T-transitions-via-coords}, 
but now putting $k=0$, one finds that 
$T_\psi(A_{\mb x,\mb y}) - A_{\mb x,\mb y} = C\cdot \one$ 
implies 
\begin{align}
C + (1-2\lambda)\beta^2 = -2\alpha\beta\cdot\Re\left( \bar\nu y_1 \right) \ .
\end{align}

Conversely, choosing $y_1$ such that this equation holds, 
and $y_k$, $k>1$, via the recurrence \eqref{eq:rec-for-ys}, 
we see that 
$(T_\psi(A_{\mb x,\mb y}) - A_{\mb x,\mb y})_{n,n+k} = (C\cdot \one)_{n,n+k}$ 
holds for all $n>0$ and $k\geq 0$. 
The equations 
$(T_\psi(A_{\mb x,\mb y}) - A_{\mb x,\mb y})_{0,k} \stackrel != 0$, 
$k>0$, 
lead to (and are solved by) the following recurrence relation for $\mb x$: 
\begin{gather}
x_k \stackrel != (\lambda\alpha^2 + (1-\lambda)\alpha)x_k 
	+ \alpha\beta\nu(-y_{k-1}) + \lambda\beta^2(x_k + y_k) 
	- \beta\bar\nu x_{k+1} \nonumber\\
\Longleftrightarrow \qquad x_{k+1} = \frac 1{\beta\bar\nu}\left( 
	(1-\lambda)(\alpha-1)x_k + \lambda\beta^2 y_k 
	- \alpha\beta\nu y_{k-1} \right) \ . \label{eq:rec-reln-for-x}
\end{gather}
The condition 
$(T_\psi(A_{\mb x,\mb y}) - A_{\mb x,\mb y})_{0,0} \stackrel != 0$ 
reads 
\begin{align}
0 \stackrel != \lambda\beta^2 - 2\beta \Re(\bar\nu x_1) \ .
\end{align}
Hence, choosing the sequences $\mb x, \mb y$ according to the 
initial conditions and recurrence relations just given, the equations 
$(T_\psi(A_{\mb x,\mb y}) - A_{\mb x,\mb y})_{n,n+k} = (C\cdot \one)_{n,n+k}$ 
are satisfied for all $n, k \geq 0$. 
The equations 
$(T_\psi(A_{\mb x,\mb y}) - A_{\mb x,\mb y})_{n+k,n} = (C\cdot \one)_{n+k,n}$ 
hold automatically, as they are the complex conjugates of the former. 

\medskip

\emph{Step 2:} 
Let $\beta^2 < \frac 1{1 + \frac{(1-2\lambda)^2}{4|\nu|^2}}$, 
fix $C > 0$, and let $\mb x, \mb y$ be as above. 
If $\frac{1-\alpha}{|\beta|} < \frac{|\nu|}{1-\lambda}$, 
then there exists $c > 0$ such that 
$\left\| T_\psi(A_{\mb x,\mb y}^{\wedge m}) - 
	A_{\mb x,\mb y}^{\wedge m} - C\cdot p_{[0,m-1]} \right\| < 2c$ 
holds for all $m \in \N$. 

\emph{Proof:} 
One finds that the matrix entries 
$\left( T_\psi(A_{\mb x,\mb y}^{\wedge m}) - 
	A_{\mb x,\mb y}^{\wedge m} - C\cdot p_{[0,m-1]} \right)_{j,k}$ 
vanish unless $j+k \in \{2m,2m+1\}$. 
If $\frac{1-\alpha}{|\beta|} < \frac{|\nu|}{1-\lambda}$, 
then not only the sequence $\mb y$ is bounded, but so 
is the sequence $\mb x$, as one 
sees by inspection of the recurrence relation \eqref{eq:rec-reln-for-x}. 
As the non-vanishing matrix entries of 
$T_\psi(A_{\mb x,\mb y}^{\wedge m}) - 
	A_{\mb x,\mb y}^{\wedge m} - C\cdot p_{[0,m-1]}$ 
are expressed in terms of simple linear combinations of 
the $x_k, y_k$, one easily sees that the matrix entries of 
$T_\psi(A_{\mb x,\mb y}^{\wedge m}) - 
	A_{\mb x,\mb y}^{\wedge m} - C\cdot p_{[0,m-1]}$ 
are also bounded by some $c > 0$. 
Since the non-vanishing matrix entries of 
$T_\psi(A_{\mb x,\mb y}^{\wedge m}) - 
	A_{\mb x,\mb y}^{\wedge m} - C\cdot p_{[0,m-1]}$ 
are concentrated on two (anti-)diagonals, the operator-norm  of 
$T_\psi(A_{\mb x,\mb y}^{\wedge m}) - 
	A_{\mb x,\mb y}^{\wedge m} - C\cdot p_{[0,m-1]}$ 
is bounded by $2c$. 

\medskip

\emph{Step 3:} If $\beta^2 < \frac 1{1 + \frac{(1-2\lambda)^2}{4|\nu|^2}}$ 
and $\frac{1-\alpha}{|\beta|} < \frac{|\nu|}{1-\lambda}$, then 
$T_\psi$ has no normal invariant state. 

\emph{Proof:} Let $\varphi$ be a normal state on $\Bl$, 
fix some $C > 0$, and let $\mb x, \mb y, c$ be as above. 
Choose $\varepsilon > 0$ such that 
$C\cdot(1-\varepsilon) > 4\sqrt\varepsilon c$, and 
$m \in \N$ such that $\varphi(p_{[0,m-1]}) > 1-\varepsilon$. 
Then, abbreviating 
$Z := T_\psi(A_{\mb x,\mb y}^{\wedge m}) - 
	A_{\mb x,\mb y}^{\wedge m} - C\cdot p_{[0,m-1]}$ 
and using $Z = Z p_{[0,m-1]}^\bot + p_{[0,m-1]}^\bot Z p_{[0,m-1]}$, 
we see 
\begin{align}
|\varphi(T_\psi(A_{\mb x,\mb y}^{\wedge m}) 
	- \varphi(A_{\mb x,\mb y}^{\wedge m})| 
	&\geq |\varphi(C \cdot p_{[0,m-1]})| 
		- |\varphi(Z)| \nonumber\\
	&> C\cdot(1-\varepsilon) 
		\underbrace{- \, |\varphi(Z p_{[0,m-1]}^\bot)|}_{\geq -
			\sqrt{\varphi(ZZ^*)\varphi(p_{[0,m-1]}^\bot)}} 
		\underbrace{- \, |\varphi(p_{[0,m-1]}^\bot Z p_{[0,m-1]})|}_{\geq -
			\sqrt{\varphi(p_{[0,m-1]}^\bot)\varphi(Z^*Z)}} 
		\nonumber\\
	&\geq C\cdot(1-\varepsilon) - 4\sqrt\varepsilon c > 0 \ . 
\end{align}
Hence, $\varphi$ cannot be invariant. 
\end{proof}

A straightforward calculation shows that 
$\beta^2 < \frac 1{1 + \frac{(1-2\lambda)^2}{4|\nu|^2}}$ 
is equivalent to 
$\frac 12 - |\frac\alpha\beta\nu| < \lambda 
	< \frac 12 + |\frac\alpha\beta\nu|$. 
It is easy to see that for $|\zeta| = 1$, the condition 
$\frac{1-\alpha}{|\beta|} < \frac{|\nu|}{1-\lambda}$ 
is equivalent to $\lambda > \frac 12(1-\alpha)$. 
Hence, as a consequence of the previous proposition and 
the non-existence statement in Proposition \ref{cor:main-result}, 
we find that for pure states $\psi$ 
with $\lambda \neq \frac 12(1\pm\alpha)$, the toy model 
transition operator $T_\psi$ admits only pure invariant 
states.

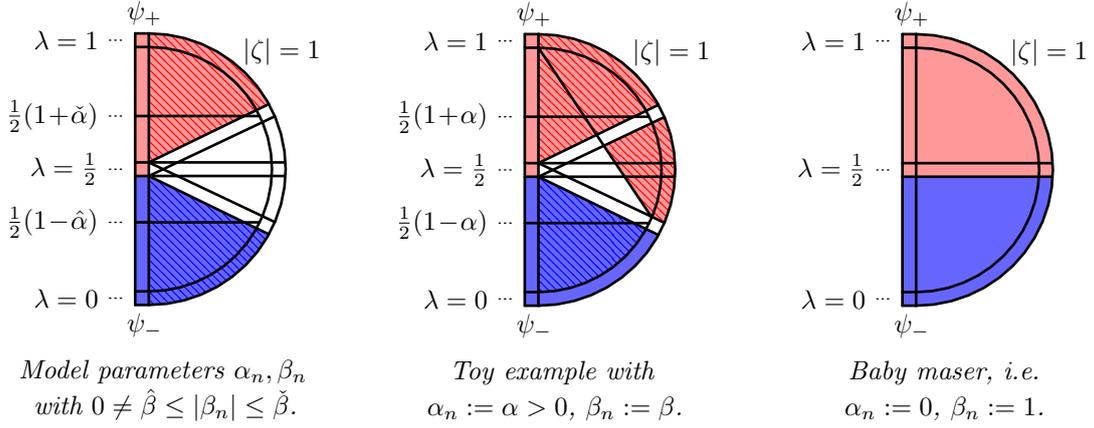
\begin{figure}[htb]
	\centering
	\begin{minipage}[c]{0.3\linewidth}
		\centering
		\begin{small}\begin{tikzpicture}[line cap=round,line join=round,>=triangle 45,x=1.8cm,y=1.8cm]
			\fill[fill=red,fill opacity=0.4] (1.9,3) -- (1.9,2.9) -- (2,2.9) -- (2,3) -- cycle; 
			\fill[fill=blue,fill opacity=0.6] (1.9,1) -- (1.9,1.1) -- (2,1.1) -- (2,1) -- cycle; 
			\fill[fill=red,fill opacity=0.4] (2,1.95) -- (1.9,1.95) -- (1.9,2.9) -- (2,2.9) -- cycle; 
			\fill[fill=blue,fill opacity=0.6] (2,1.95) -- (1.9,1.95) -- (1.9,1.1) -- (2,1.1) -- cycle; 
			\fill [shift={(2,2)},fill=red,fill opacity=.4]  
			(0,.05) --  plot[domain=0.5:1.57,variable=\t]({1*1*cos(\t r)+0*1*sin(\t r)},{0*1*cos(\t r)+1*1*sin(\t r)}) -- cycle; 
			\fill [shift={(2,2)},pattern=north west lines,pattern color=red]  
			(0,.05) --  plot[domain=0.5:1.57,variable=\t]({1*1*cos(\t r)+0*1*sin(\t r)},{0*1*cos(\t r)+1*1*sin(\t r)}) -- cycle; 
			\fill [shift={(2,2)},fill=blue,fill opacity=0.6]  plot[domain=-0.34:-1.57,variable=\t](0,-.05) -- plot[domain=-1.57:-0.5,variable=\t]({1*1*cos(\t r)+0*1*sin(\t r)},{0*1*cos(\t r)+1*1*sin(\t r)}) -- cycle;
			\fill [shift={(2,2)},pattern=north west lines,pattern color=blue]  plot[domain=-0.34:-1.57,variable=\t](0,-.05) -- plot[domain=-1.57:-0.5,variable=\t]({1*1*cos(\t r)+0*1*sin(\t r)},{0*1*cos(\t r)+1*1*sin(\t r)}) -- cycle;

			\draw [line width=1pt] (1.9,2.05)-- (2.99,2.05); 
			\draw [line width=1pt] (1.9,1.95)-- (2.99,1.95); 
			\draw [line width=1pt] (1.9,{2 + .9*sin(-0.45 r)})-- ({2 + .9*cos(-0.45 r)},{2 + .9*sin(-0.45 r)}); 
			\draw [line width=1pt] (1.9,{2 + .9*sin(0.45 r)})-- ({2 + .9*cos(0.45 r)},{2 + .9*sin(0.45 r)}); 
			\draw [line width=1pt] (1.9,3)-- (2,3); 
			\draw [line width=1pt] (1.9,2.9)-- (2,2.9); 
			\draw [line width=1pt] (1.9,1.1)-- (2,1.1); 
			\draw [line width=1pt] (1.9,1)-- (2,1); 
			\draw [line width=1pt] (1.9,3)-- (1.9,1); 
			\draw [line width=1pt] (2,1)-- (2,3); 
			\draw [shift={(2,2)},line width=1pt]  plot[domain=-1.57:1.57,variable=\t]({1*1*cos(\t r)+0*1*sin(\t r)},{0*1*cos(\t r)+1*1*sin(\t r)}); 
			\draw [shift={(2,2)},line width=1pt]  plot[domain=-1.57:1.57,variable=\t]({1*0.9*cos(\t r)+0*0.9*sin(\t r)},{0*0.9*cos(\t r)+1*0.9*sin(\t r)}); 
			\draw [line width=1pt] (2,2.05)-- ({2 + cos(-0.4 r)},{2 + sin(-0.4 r)});
			\draw [line width=1pt] (2,1.95)-- ({2 + cos(-0.5 r)},{2 + sin(-0.5 r)});
			\draw [line width=1pt] (2,1.95)-- ({2 + cos(0.4 r)},{2 + sin(0.4 r)});
			\draw [line width=1pt] (2,2.05)-- ({2 + cos(0.5 r)},{2 + sin(0.5 r)});

			\draw (1.7,2.95) node[anchor=east] {$\lambda=1$};
			\draw [dash pattern=on .5pt off 1.5pt,line width=.5pt] (1.8,2.95)-- (1.7,2.95);
			\draw (1.7,{2 + .9*sin(-0.45 r)}) node[anchor=east] {$\frac12(1\!-\!\hat\alpha)$};
			\draw [dash pattern=on .5pt off 1.5pt,line width=.5pt] (1.8,{2 + .9*sin(-0.45 r)})-- (1.7,{2 + .9*sin(-0.45 r)});
			\draw (1.7,{2 + .9*sin(0.45 r)}) node[anchor=east] {$\frac12(1\!+\!\check\alpha)$};
			\draw [dash pattern=on .5pt off 1.5pt,line width=.5pt] (1.8,{2 + .9*sin(0.45 r)})-- (1.7,{2 + .9*sin(0.45 r)});
			\draw (1.7,2) node[anchor=east] {$\lambda=\frac12$};
			\draw [dash pattern=on .5pt off 1.5pt,line width=.5pt] (1.8,2)-- (1.7,2);
			\draw (1.7,1.05) node[anchor=east] {$\lambda=0$};
			\draw [dash pattern=on .5pt off 1.5pt,line width=.5pt] (1.8,1.05)-- (1.7,1.05);
			\draw (2.62,2.7) node[anchor=south west] {$|\zeta|=1$};
			\draw (1.97,2.99) node[anchor=south] {$\psi_+$};
			\draw (1.97,1.01) node[anchor=north] {$\psi_-$};
		\end{tikzpicture}\\
		\emph{Model parameters $\alpha_n,\beta_n$\\ with 
		$0\neq\hat\beta\leq|\beta_n|\leq\check\beta$.}\end{small}
	\end{minipage}\hspace{0.04\linewidth}
	\begin{minipage}[c]{0.3\linewidth}
		\centering
		\begin{small}\begin{tikzpicture}[line cap=round,line join=round,>=triangle 45,x=1.8cm,y=1.8cm]
			\fill[fill=red,fill opacity=0.4] (1.9,3) -- (1.9,2.9) -- (2,2.9) -- (2,3) -- cycle; 
			\fill[fill=blue,fill opacity=0.6] (1.9,1) -- (1.9,1.1) -- (2,1.1) -- (2,1) -- cycle; 
			\fill[fill=red,fill opacity=0.4] (2,1.95) -- (1.9,1.95) -- (1.9,2.9) -- (2,2.9) -- cycle; 
			\fill[fill=blue,fill opacity=0.6] (2,1.95) -- (1.9,1.95) -- (1.9,1.1) -- (2,1.1) -- cycle; 
			\fill [shift={(2,2)},fill=red,fill opacity=.4]  plot[domain=0.58:1,variable=\s]({(1-\s)*0+\s*.9*cos(-.4 r)},{(1-\s)*.9+\s*.9*sin(-.4 r)}) --  plot[domain=-0.4:0.4,variable=\t]({1*1*cos(\t r)+0*1*sin(\t r)},{0*1*cos(\t r)+1*1*sin(\t r)}) -- cycle; 
			\fill [shift={(2,2)},pattern=north west lines,pattern color=red]  plot[domain=0.58:1,variable=\s]({(1-\s)*0+\s*.9*cos(-.4 r)},{(1-\s)*.9+\s*.9*sin(-.4 r)}) --  plot[domain=-0.4:0.4,variable=\t]({1*1*cos(\t r)+0*1*sin(\t r)},{0*1*cos(\t r)+1*1*sin(\t r)}) -- cycle; 
			\fill [shift={(2,2)},fill=red,fill opacity=.4]  
			(0,.05) --  plot[domain=0.5:1.57,variable=\t]({1*1*cos(\t r)+0*1*sin(\t r)},{0*1*cos(\t r)+1*1*sin(\t r)}) -- cycle; 
			\fill [shift={(2,2)},pattern=north west lines,pattern color=red]  
			(0,.05) --  plot[domain=0.5:1.57,variable=\t]({1*1*cos(\t r)+0*1*sin(\t r)},{0*1*cos(\t r)+1*1*sin(\t r)}) -- cycle; 
			\fill [shift={(2,2)},fill=blue,fill opacity=0.6]  plot[domain=-0.5:-1.57,variable=\t]({1*0.9*cos(\t r)+0*0.9*sin(\t r)},{0*0.9*cos(\t r)+1*0.9*sin(\t r)}) -- plot[domain=-1.57:-0.5,variable=\t]({1*1*cos(\t r)+0*1*sin(\t r)},{0*1*cos(\t r)+1*1*sin(\t r)}) -- cycle; 
			\fill [shift={(2,2)},fill=blue,fill opacity=0.6]  plot[domain=-0.34:-1.57,variable=\t](0,-.05) -- plot[domain=-1.57:-0.5,variable=\t]({1*0.9*cos(\t r)+0*0.9*sin(\t r)},{0*0.9*cos(\t r)+1*0.9*sin(\t r)}) -- cycle; 
			\fill [shift={(2,2)},pattern=north west lines,pattern color=blue]  plot[domain=-0.34:-1.57,variable=\t](0,-.05) -- plot[domain=-1.57:-0.5,variable=\t]({1*0.9*cos(\t r)+0*0.9*sin(\t r)},{0*0.9*cos(\t r)+1*0.9*sin(\t r)}) -- cycle; 

			\draw [line width=1pt] (1.9,2.05)-- (2.99,2.05); 
			\draw [line width=1pt] (1.9,1.95)-- (2.99,1.95); 
			\draw [line width=1pt] (1.9,{2 + .9*sin(-0.45 r)})-- ({2 + .9*cos(-0.45 r)},{2 + .9*sin(-0.45 r)}); 
			\draw [line width=1pt] (1.9,{2 + .9*sin(0.45 r)})-- ({2 + .9*cos(0.45 r)},{2 + .9*sin(0.45 r)}); 
			\draw [line width=1pt] (1.9,3)-- (2,3); 
			\draw [line width=1pt] (1.9,2.9)-- (2,2.9); 
			\draw [line width=1pt] (1.9,1.1)-- (2,1.1); 
			\draw [line width=1pt] (1.9,1)-- (2,1); 
			\draw [line width=1pt] (1.9,3)-- (1.9,1); 
			\draw [line width=1pt] (2,1)-- (2,3); 
			\draw [shift={(2,2)},line width=1pt]  plot[domain=-1.57:1.57,variable=\t]({1*1*cos(\t r)+0*1*sin(\t r)},{0*1*cos(\t r)+1*1*sin(\t r)}); 
			\draw [shift={(2,2)},line width=1pt]  plot[domain=-1.57:1.57,variable=\t]({1*0.9*cos(\t r)+0*0.9*sin(\t r)},{0*0.9*cos(\t r)+1*0.9*sin(\t r)}); 
			\draw [line width=1pt] (2,2.05)-- ({2 + cos(-0.4 r)},{2 + sin(-0.4 r)});
			\draw [line width=1pt] (2,1.95)-- ({2 + cos(-0.5 r)},{2 + sin(-0.5 r)});
			\draw [line width=1pt] (2,1.95)-- ({2 + cos(0.4 r)},{2 + sin(0.4 r)});
			\draw [line width=1pt] (2,2.05)-- ({2 + cos(0.5 r)},{2 + sin(0.5 r)});
			\draw [line width=1pt] (2,2.9)-- ({2 + .9*cos(-0.4 r)},{2 + .9*sin(-0.4 r)});

			\draw (1.7,2.95) node[anchor=east] {$\lambda=1$};
			\draw [dash pattern=on .5pt off 1.5pt,line width=.5pt] (1.8,2.95)-- (1.7,2.95);
			\draw (1.7,2) node[anchor=east] {$\lambda=\frac12$};
			\draw [dash pattern=on .5pt off 1.5pt,line width=.5pt] (1.8,2)-- (1.7,2);
			\draw (1.7,{2 + .9*sin(-0.45 r)}) node[anchor=east] {$\frac12(1\!-\!\alpha)$};
			\draw [dash pattern=on .5pt off 1.5pt,line width=.5pt] (1.8,{2 + .9*sin(-0.45 r)})-- (1.7,{2 + .9*sin(-0.45 r)});
			\draw (1.7,{2 + .9*sin(0.45 r)}) node[anchor=east] {$\frac12(1\!+\!\alpha)$};
			\draw [dash pattern=on .5pt off 1.5pt,line width=.5pt] (1.8,{2 + .9*sin(0.45 r)})-- (1.7,{2 + .9*sin(0.45 r)});
			\draw (1.7,1.05) node[anchor=east] {$\lambda=0$};
			\draw [dash pattern=on .5pt off 1.5pt,line width=.5pt] (1.8,1.05)-- (1.7,1.05);
			\draw (2.62,2.7) node[anchor=south west] {$|\zeta|=1$};
			\draw (1.97,2.99) node[anchor=south] {$\psi_+$};
			\draw (1.97,1.01) node[anchor=north] {$\psi_-$};
		\end{tikzpicture}\\
		\emph{Toy example with\\ $\alpha_n:=\alpha>0$, $\beta_n:=\beta$.}\end{small}
	\end{minipage}\hspace{0.04\linewidth}
	\begin{minipage}[c]{0.3\linewidth}
		\centering
		\begin{small}\begin{tikzpicture}[line cap=round,line join=round,>=triangle 45,x=1.8cm,y=1.8cm]
			\fill[fill=red,fill opacity=0.4] (1.9,3) -- (1.9,2.9) -- (2,2.9) -- (2,3) -- cycle; 
			\fill[fill=blue,fill opacity=0.6] (1.9,1) -- (1.9,1.1) -- (2,1.1) -- (2,1) -- cycle; 
			\fill[fill=red,fill opacity=0.4] (2,1.95) -- (1.9,1.95) -- (1.9,2.9) -- (2,2.9) -- cycle; 
			\fill[fill=blue,fill opacity=0.6] (2,1.95) -- (1.9,1.95) -- (1.9,1.1) -- (2,1.1) -- cycle; 
			\fill [shift={(2,2)},fill=red,fill opacity=0.4]  (0,-0.05) --  plot[domain=-0.05:1.57,variable=\t]({1*0.9*cos(\t r)+0*0.9*sin(\t r)},{0*0.9*cos(\t r)+1*0.9*sin(\t r)}) -- cycle; 
			\fill [shift={(2,2)},fill=blue,fill opacity=0.6]  (0,-0.05) --  plot[domain=-1.57:-0.05,variable=\t]({1*0.9*cos(\t r)+0*0.9*sin(\t r)},{0*0.9*cos(\t r)+1*0.9*sin(\t r)}) -- cycle; 
			\fill [shift={(2,2)},fill=red,fill opacity=0.4]  plot[domain=1.57:-0.05,variable=\t]({1*0.9*cos(\t r)+0*0.9*sin(\t r)},{0*0.9*cos(\t r)+1*0.9*sin(\t r)}) --  plot[domain=-0.05:1.57,variable=\t]({1*1*cos(\t r)+0*1*sin(\t r)},{0*1*cos(\t r)+1*1*sin(\t r)}) -- cycle; 
			\fill [shift={(2,2)},fill=blue,fill opacity=0.6]  plot[domain=-0.05:-1.57,variable=\t]({1*0.9*cos(\t r)+0*0.9*sin(\t r)},{0*0.9*cos(\t r)+1*0.9*sin(\t r)}) --  plot[domain=-1.57:-0.05,variable=\t]({1*1*cos(\t r)+0*1*sin(\t r)},{0*1*cos(\t r)+1*1*sin(\t r)}) -- cycle; 

			\draw [line width=1pt] (1.9,2.05)-- (2.99,2.05); 
			\draw [line width=1pt] (1.9,1.95)-- (2.99,1.95); 
			\draw [line width=1pt] (1.9,3)-- (2,3); 
			\draw [line width=1pt] (1.9,2.9)-- (2,2.9); 
			\draw [line width=1pt] (1.9,1.1)-- (2,1.1); 
			\draw [line width=1pt] (1.9,1)-- (2,1); 
			\draw [line width=1pt] (1.9,3)-- (1.9,1); 
			\draw [line width=1pt] (2,1)-- (2,3); 
			\draw [shift={(2,2)},line width=1pt]  plot[domain=-1.57:1.57,variable=\t]({1*1*cos(\t r)+0*1*sin(\t r)},{0*1*cos(\t r)+1*1*sin(\t r)}); 
			\draw [shift={(2,2)},line width=1pt]  plot[domain=-1.57:1.57,variable=\t]({1*0.9*cos(\t r)+0*0.9*sin(\t r)},{0*0.9*cos(\t r)+1*0.9*sin(\t r)}); 

			\draw (1.7,2.95) node[anchor=east] {$\lambda=1$};
			\draw [dash pattern=on .5pt off 1.5pt,line width=.5pt] (1.8,2.95)-- (1.7,2.95);
			\draw (1.7,2) node[anchor=east] {$\lambda=\frac12$};
			\draw [dash pattern=on .5pt off 1.5pt,line width=.5pt] (1.8,2)-- (1.7,2);
			\draw (1.7,1.05) node[anchor=east] {$\lambda=0$};
			\draw [dash pattern=on .5pt off 1.5pt,line width=.5pt] (1.8,1.05)-- (1.7,1.05);
			\draw (2.62,2.7) node[anchor=south west] {$|\zeta|=1$};
			\draw (1.97,2.99) node[anchor=south] {$\psi_+$};
			\draw (1.97,1.01) node[anchor=north] {$\psi_-$};
		\end{tikzpicture}\\
		\emph{Baby maser, \ie\\ $\alpha_n:=0$, $\beta_n:=1$.}\end{small}
	\end{minipage}
	\caption[Invariant normal states]{Blue areas indicate that there is an invariant normal state for $T_\psi$, red areas indicate that there is no such state. Hatched areas are new compared with \cite[Fig.\,6.2]{BGKRSS13}.
	In case of the toy example, there are no \emph{pure} invariant states at 
	$\lambda = \frac 12(1 \pm \alpha)$. 
	Other regions for which we do not have any results are left blank.
	}\label{fig:summary}
\end{figure}

\medskip

To summarize, the regions of the parameter space for the \enquote{atomic} state $\psi$ where $T_\psi$ admits an {invariant normal state} 
and those where it does not are shown in Figure \ref{fig:summary}.

\subsection{Example: one-atom maser with random interaction time}

The evolution of an electromagnetic mode inside a 
perfect (no energy loss) cavity which interacts sequentially, 
according to the Jaynes-Cummings model, during time-intervals of length $\tau > 0$ 
with two-level atoms prepared in a state $\psi$ 
is described by the generalised one-atom maser 
with parameters $\alpha_n,\beta_n$ chosen as 
(cf. \cite[Ex.\,3.4]{BGKRSS13})
\begin{align}
\alpha_n(\tau) = \cos(g\tau\sqrt{n}) \ , \qquad 
\beta_n(\tau) = -\sin(g\tau\sqrt{n}) \ .
\end{align}
To make the dependence on $\tau$ explicit, we denote the 
corresponding transition operator by $T_{\psi,\tau}$. 
As shown in \cite[Thm.\,3.3]{bruneau09}, this model admits invariant states 
if $\psi$ is a \enquote{thermal state} (i.e. $\nu = 0$) 
with $\lambda < \frac 12$. 
If $\beta_n(\tau) \neq 0$ for all $n > 0$, referred to as 
the \enquote{non-resonant} case in \cite{bruneau09}, then 
the invariant state is unique, and is given by 
$\varphi(e_{n,m}) = \delta_{n,m}\cdot\frac{1-2\lambda}{1-\lambda} 
	\left(\frac\lambda{1-\lambda}\right)^n$. 
In this situation, the invariant state is also \emph{absorbing}, 
meaning that for all normal states $\theta$ and observables $x \in \Bl$ 
one has $\lim_n \theta(T_{\psi,\tau}^n(x)) = \varphi(x)$; 
see \cite[Thm.\,3.2]{bruneau13} (the notions 
\enquote{mixing}, defined there, and \enquote{absorbing} coincide for faithful states). 
However, the model exhibits pathological behaviour: 
For generic $\tau > 0$, the values of 
$|\beta_n(\tau)|$ come arbitrarily close to $0$. 
Nesting intervals, one can show: 
\begin{itemize}
\item For a fixed 
\emph{pure} atomic state $\psi$ with $\lambda < \frac 12$, 
the set of values of $\tau$, for which $T_{\psi,\tau}$ 
admits no \emph{pure} invariant state, is dense in $(0,\infty)$. 
\item For a fixed 
\emph{diagonal} atomic state $\psi$ with $\lambda < \frac 12$, 
the set of values of $\tau$, for which there is a normal state 
$\tilde\varphi$ with $\supp\tilde\varphi\leq p_{[0,m]}$ 
for some $m \in \N_0$ 
and $((T_{\psi,\tau})_*^n(\tilde\varphi))_{n\in\N_0}$ 
converging slower than geometrically towards the invariant state, 
is dense in $(0,\infty)$. 
\end{itemize}
It was suggested in \cite[Open\,problem\,3]{BJM13} that adding some 
randomness might cure the problems. 
Therefore, we consider here the model where 
$\tau$ is randomly distributed according to 
some probability measure $\rho$ on 
$[0,\infty)$. 
The evolution of the electromagnetic mode is then described by 
the ucp-map 
\begin{align}
T_{\psi,\rho}(x) := \int_{[0,\infty)} T_{\psi,\tau}(x) \, d\rho(\tau) \ , 
	\qquad x \in \Bl \ .
\end{align}
For a Hilbert space $\H$, denote by $\UCP(\BH)$ the 
set of ucp-maps on $\BH$. 
The following lemma will be useful: 

\begin{lemma}\label{lem:subharmonics-for-mixtures}
Let $\mu$ be a probability measure on a measure space $\Omega$, 
let $\H$ be a separable Hilbert space, 
and let $T: \Omega \to \UCP(\BH)$, $\omega \mapsto T_\omega$ 
be a function such that $\omega \mapsto \langle\xi,\, T_\omega(x)\xi\rangle$ 
is measurable for each $\xi \in \H$ and $x \in \BH$. 
A projection $p \in \BH$ is subharmonic for 
$\int_\Omega T_\omega \, d\mu(\omega)$ if and only if 
$p$ is subharmonic for $T_\omega$ for almost all $\omega \in \Omega$. 
\end{lemma}
\begin{proof}
Recall that a projection $p$ is subharmonic for a ucp-map $T$ 
if and only if $T(p) \geq p$, if and only if 
$T(p^\bot) \leq p^\bot$, if and only if $pT(p^\bot)p = 0$. 
For each $\xi \in \H$ we have 
$\langle p\xi,\,\int_\Omega T_\omega(p^\bot) \, d\mu(\omega) p \xi \rangle
\stackrel{\text{(by def)}}= \int_\Omega \langle\xi,\, pT_\omega(p^\bot)p \xi \rangle \, d\mu(\omega) = 0$ 
if and only if $pT_\omega(p^\bot)p \xi = 0$ for almost all $\omega$, 
since $pT_\omega(p^\bot)p \geq 0$. 
As $\H$ is separable, we have 
$p\int_\Omega T_\omega(p^\bot) \, d\mu(\omega) p = 0$ 
if and only if $pT_\omega(p^\bot)p = 0$ for almost all $\omega$. 
\end{proof}

\begin{prop}
\label{prop:random-JC}
Let $\rho = D_\rho(\tau)d\tau$ be a probability measure 
with density $D_\rho \in C^1([0,\infty))$ such that 
$\frac{d}{d\tau}D_\rho \in L^1([0,\infty))$. Let 
$0\leq \lambda \leq 1$. Then: 
\begin{enumerate}
\item $T_{\psi,\rho}$ has a normal invariant state if $\lambda < \frac 12$; 
	if $\lambda > \frac 12$, then $T_{\psi,\rho}$ has no normal invariant state. 
\item If $\lambda < \frac 12$ and $\psi$ is faithful, 
	then the normal invariant state for $T_{\psi,\rho}$ is absorbing. 
\end{enumerate}
\end{prop}

\begin{proof}
First, the parameters $\mu_n,\lambda_n,\eta_n$ of $T_{\psi,\rho}$ 
are given by 
\begin{align}\label{eq:bnd-rates-for-smeared-mm}
\mu_n &= \int_{(0,\infty)} \Tr(e_{n,n}T_{\psi,\tau}(e_{n-1,n-1})) \, d\rho(\tau) 
	= (1-\lambda) \int_{(0,\infty)} \sin^2(g\tau\sqrt n) \, d\rho(\tau) \ , \nonumber\\
\lambda_n &= \int_{(0,\infty)} \Tr(e_{n,n}T_{\psi,\tau}(e_{n+1,n+1})) \, d\rho(\tau) 
	= \lambda \int_{(0,\infty)} \sin^2(g\tau\sqrt{n+1}) \, d\rho(\tau) \ , \\
\eta_n &= \int_{(0,\infty)} \Tr(e_{n,n+1}T_{\psi,\tau}(e_{n,n})) \, d\rho(\tau) 
	= \bar\nu \int_{(0,\infty)} \sin(g\tau\sqrt{n+1})\cos(g\tau\sqrt{n+1}) \, d\rho(\tau) \nonumber\\
	&= \frac{\bar\nu}2 \int_{(0,\infty)} \sin(2g\tau\sqrt{n+1}) \, d\rho(\tau) \ , \nonumber
\end{align}
cf. Figure \ref{fig:action}. 
As $n \to \infty$, they approach the limits $\lim_n\mu_n = \frac 12(1-\lambda)$, 
$\lim_n \lambda_n = \frac\lambda 2$ 
and $\lim_n \eta_n = 0$: for example, 
\begin{align}
\left| \frac 2{\bar\nu}\eta_n \right| = 
&\left|\int_{(0,\infty)} \sin(2g\tau\sqrt{n+1}) \, d\rho(\tau) \right|
	\leq \frac 1{2g\sqrt{n+1}}\biggl\{ 
	\underbrace{\left|\left.\cos(2g\tau\sqrt{n+1})D_\rho(\tau)\right|_{0}^\infty\right|}_{\leq D_\rho(0)} \nonumber\\
	&\hspace{6em}+ \underbrace{\int_{(0,\infty)} |\cos(2g\tau\sqrt{n+1}) (\frac{d}{d\tau}D_\rho)(\tau)|\,d\tau}_{\leq \|\frac{d}{d\tau}D_\rho\|_1} 
	\biggr\} 
	\quad\stackrel{n\to\infty}\longrightarrow\quad 0 \ .
\end{align}
An application of Theorem \ref{thm:main-result} along the 
lines of the proof of Proposition \ref{cor:main-result} 
shows part 1. 

\medskip

For part 2, observe that $T_{\psi,\rho}$ is irreducible by 
Lemma \ref{lem:subharmonics-for-mixtures}, since 
almost all $T_{\psi,\tau}$ are irreducible. 
Hence, the invariant state for $T_{\psi,\rho}$, 
existing thanks to part 1, is faithful. 
Moreover, if $\psi_-$ is the state on $M_2$ corresponding 
to $\lambda = 0$, then $T_{\psi,\rho}$ is a non-trivial convex combination 
of $T_{\psi_-,\rho}$ with some other ucp-map. 
As the invariant state $x \mapsto \skapro{xe_0,e_0}$ 
of $T_{\psi_-,\rho}$ is absorbing, 
$T_{\psi,\rho}$ has an absorbing state by Theorem \ref{thm:convex-comb} 
or by Proposition \ref{prop:convex-comb}. 
\end{proof}

\appendix

\subsection*{Appendix: Approach to equilibrium}

A normal invariant state $\varphi$ 
is called \emph{absorbing} for a ucp-map $T$, if for each normal 
state $\theta$ and for all observables $x$ 
we have $\lim_n \theta(T^n(x)) = \varphi(x)$. 
The following result is used in the proof of 
Proposition \ref{prop:random-JC} above: 

\begin{thm}[{\cite[Kor.\,2.2.11\,\&\,Satz\,2.4.9]{Haag06}}]
\label{thm:convex-comb}
Let $R, S$ be ucp-maps on $\B(\H)$, and 
let $T = \lambda R + (1-\lambda) S$ for some 
$0 < \lambda \leq 1$. Suppose that $T$ 
admits a faithful invariant normal state. 
Then if $R$ has an absorbing state, so does $T$. 
\end{thm}

As the proof given in \cite{Haag06} is available in German language 
only, and as the author is unaware of another reference 
for this statement, we will here give a quantitative 
variant of Theorem \ref{thm:convex-comb}. 
This allows to give a slightly shorter proof, is 
sufficient for the purposes above, and might be of 
independent interest. 

Let $g$ be a function $\N \to \R$. 
A function $f: \N \to \R$ is said to be of order $g$, written 
$f \in \O(g)$, if there is $C > 0$ such that 
$|f(n)| < C g(n)$ for all $n \in \N$. 

\begin{prop}
\label{prop:convex-comb}
Let $R, S$ be ucp-maps on $\Bl$, and 
let $T = \lambda R + (1-\lambda) S$ for some 
$0 < \lambda < 1$. Moreover, suppose that $T$ 
admits a faithful normal invariant state $\varphi$ 
of exponential fall-off: 
$\sum_{k\geq n}\varphi(e_{k,k}) \in \O(e^{-\gamma_2n})$. 
Then if $R$ has an absorbing state $\varphi_R$ such that 
$\|\theta\circ R^n - \varphi_R \circ \textnormal{Ad}_{p_{[0,m]}}\| 
	\leq e^{\gamma_0m - \gamma_1n}$ 
holds for all normal states $\theta$ with $\supp\theta \leq p_{[0,m]}$, $m \in \N$, 
then $T$ has an absorbing state $\varphi_T$ such that 
$\|\theta\circ T^n - \varphi_T\| \in \O(n^{-\gamma(a)})$, 
where 
$\gamma(a) = \frac{\gamma_1\gamma_2}{-a\ln\lambda\cdot(\gamma_0 + \gamma_2)}$, 
holds for all $a > 1$ and for all 
$\theta$ with $\supp\theta \leq p_{[0,m]}$ for some $m \in \N$. 
\end{prop}
\begin{proof}(along the lines of the proof of \cite[Satz\,2.4.9]{Haag06})
The idea is to show that 
$\sup_{N_1,N_2\geq N} \left\| \theta \circ T^{N_1} - \theta \circ T^{N_2} \right\| 
	\in \O(N^{-\gamma(a)})$. 

For words $\vec i$ over $\{0,1\}$, recursively define 
ucp-maps $T^{\vec i}$ by putting 
$T^\emptyset := \id$, $T^{\vec i 0} := T^{\vec i} \circ R$ 
and $T^{\vec i 1} := T^{\vec i} \circ S$. 
Let $\mu$ be the probability measure on $\{0,1\}^N$ with 
$\mu(\{\vec i\}) := \lambda^n(1-\lambda)^{N-n}$, if $\vec i$ is 
a word of length $N$ containing the letter \enquote{$0$} $n$-times. 
Then we have 
$\theta \circ T^N = \sum_{\vec i \in \{0,1\}^N} 
	\mu(\{\vec i\}) \cdot \theta\circ T^{\vec i}$. 
Denote by $\mathfrak R^N_r \subseteq \{0,1\}^n$ the set of 
words containing at least one run of $r$ consecutive $0$'s. 
We have $1-\mu(\mathfrak R^N_r) \leq (1-\lambda^r)^{\lfloor \frac Nr \rfloor}$ 
(divide $\{1,2,\ldots,N\}$ into $\lfloor \frac Nr \rfloor$ 
blocks of length $r$ and only count runs fitting into 
one of these blocks). 
Hence, for $N,N_1,N_2 \in \N$ with $N_1,N_2 \geq N$, 
\begin{align}
\left\| \theta \circ T^{N_1} - \theta \circ T^{N_2} \right\| &\leq 
\left\| \sum_{\vec i \in \mathfrak R^{N}_r} \mu(\{\vec i\}) 
	\left(\theta \circ T^{N_1-N} \circ T^{\vec i} 
	- \theta \circ T^{N_2-N} \circ T^{\vec i} \right) \right\| \nonumber\\
	&\quad + \underbrace{\left\| \sum_{\vec i \notin \mathfrak R^{N}_r} 
		\mu(\{\vec i\}) \left( \theta \circ T^{N_1-N} \circ T^{\vec i} 
	- \theta \circ T^{N_2-N} \circ T^{\vec i} \right) \right\|}_{\leq 2(1-\lambda^r)^{\lfloor \frac Nr \rfloor}} \ .
\end{align}
Let $\theta$ be a normal state with $\supp\theta \leq p_{[0,m]}$. 
Since $\varphi$ is faithful, there exists $C_1 > 0$ with 
$\theta \leq C_1 \varphi$. 
Since $\varphi$ falls off exponentially, there are 
$C_2,\gamma_2 > 0$ such that $\varphi(p_{[0,M]}^\bot) < C_2^2 e^{-2\gamma_2 M}$ 
for all $M \in \N$. 
Then we have, for $x \in \Bl$ and $M \in \N$, 
\begingroup
\allowdisplaybreaks
\begin{align}
&\left| \sum_{\vec i \in \mathfrak R^{N}_r} \mu(\{\vec i\}) 
	\left(\theta \circ T^{N_1-N} \circ T^{\vec i} 
	- \theta \circ T^{N_2-N} \circ T^{\vec i} \right)(x) \right| \nonumber\\
&\quad = \left| \sum_{k=r}^N \sum_{\vec i \in \{0,1\}^{k-r}} \mu(\{\vec i\}) 
	\left(\theta \circ T^{N_1-k} \circ R^r \circ T^{\vec i} 
	- \theta \circ T^{N_2-k} \circ R^r \circ T^{\vec i} \right)(x) \right| \nonumber\\
&\quad \leq \sum_{k=r}^N \sum_{\vec i \in \{0,1\}^{k-r}} \mu(\{\vec i\}) \biggl( 
	\left| \left(\theta \circ T^{N_1-k} 
		- \theta \circ T^{N_2-k} \right)(p_{[0,M]}R^r(T^{\vec i}(x))p_{[0,M]}) \right| 
	\nonumber\\
	&\qquad +\left| \left(\theta \circ T^{N_1-k} 
		- \theta \circ T^{N_2-k} \right)(p_{[0,M]}^\bot R^r(T^{\vec i}(x))p_{[0,M]}) \right| 
	\nonumber\\
	&\qquad + \left| \left(\theta \circ T^{N_1-k} 
		- \theta \circ T^{N_2-k} \right)(p_{[0,M]}R^r(T^{\vec i}(x))p_{[0,M]}^\bot) \right| 
	\nonumber\\
	&\qquad +  \underbrace{\left| \left(\theta \circ T^{N_1-k} 
		- \theta \circ T^{N_2-k} \right)(p_{[0,M]}^\bot R^r(T^{\vec i}(x))p_{[0,M]}^\bot) \right|}_{\stackrel{(*)}{\leq} 2C_1 C_2e^{-\gamma_2M}\|x\|} \biggr) \nonumber\\
	&\quad\leq \sum_{k=r}^N \sum_{\vec i \in \mathfrak R^{k}_r} \mu(\{\vec i\}) \biggl( 
	\underbrace{\left| \left(\theta \circ T^{N_1-k} - \varphi_R \right)(p_{[0,M]}R^r(T^{\vec i}(x))p_{[0,M]}) \right|}_{\leq e^{\gamma_0M - \gamma_1r} \|x\|} \nonumber\\
		&\qquad+ \underbrace{\left| \left(\theta \circ T^{N_2-k} - \varphi_R \right)(p_{[0,M]}R^r(T^{\vec i}(x))p_{[0,M]}) \right|}_{\leq e^{\gamma_0M - \gamma_1r} \|x\|} \biggr) 
		+ 6C_1 C_2e^{-\gamma_2M}\|x\| \nonumber\\
	&\quad\leq (2C^Mg(r) + 6C_1 C_2e^{-\gamma_2M}) \|x\| \ .
\end{align}
\endgroup
where for (*) the Cauchy-Schwarz inequality was used. 
Altogether, we have 
\begin{align}
\left\| \theta \circ T^{N_1} - \theta \circ T^{N_2} \right\| 
	&\leq 2C e^{\gamma_0M - \gamma_1r} + 6C_1 C_2e^{-\gamma_2M} 
		+ 2(1-\lambda^r)^{\lfloor \frac Nr \rfloor} \ .
\end{align}
Now, choosing $r := \lfloor \frac{\ln N}{-a\ln\lambda} \rfloor$ 
for some $a>1$, one finds 
$(1-\lambda^r)^{\lfloor \frac Nr \rfloor} 
	= \exp(\lfloor \frac Nr \rfloor \cdot \ln(1-\lambda^r))
	\leq \exp(- \lfloor \frac Nr \rfloor \cdot \lambda^r) 
	\in \O(\exp( \frac N{a\ln\lambda} e^{\ln\lambda\frac{\ln N}{-a\ln\lambda}})) 
	= \O(\exp(\frac{N^{1-\frac 1a}}{a\ln\lambda}))$. 
Moreover, putting 
$M := \lfloor \frac{\gamma_1}{\gamma_0 + \gamma_2} \frac{\ln N}{-a\ln\lambda} \rfloor$, 
we obtain 
$e^{\gamma_0M - \gamma_1r}, e^{-\gamma_2M} \in 
	\O(\exp(-\frac{\gamma_1\gamma_2}{\gamma_0 + \gamma_2} \frac{\ln N}{-a\ln\lambda})) 
	= \O(N^{-\gamma(a)})$, 
where 
$\gamma(a) = \frac{\gamma_1\gamma_2}{-a\ln\lambda\cdot(\gamma_0 + \gamma_2)}$. 
\end{proof}

\newcommand{\etalchar}[1]{$^{#1}$}
\renewcommand{\refname}{References}
\providecommand{\bysame}{\leavevmode\hbox to3em{\hrulefill}\thinspace}
\providecommand{\MR}{\relax\ifhmode\unskip\space\fi MR }
\providecommand{\MRhref}[2]{%
  \href{http://www.ams.org/mathscinet-getitem?mr=#1}{#2}
}
\providecommand{\href}[2]{#2}

\end{document}